\documentclass[lettersize]{IEEEtran}
\IEEEoverridecommandlockouts
\usepackage{lineno}
\usepackage{hyperref}
\usepackage{cite}
\usepackage{amsmath,amssymb,amsfonts}
\usepackage{amsmath}
\usepackage{amsthm}

\usepackage{graphicx}
\usepackage{textcomp}
\usepackage{xcolor}
\usepackage{graphicx}
\usepackage{float}
\usepackage{amsmath,mleftright}
\usepackage{amsfonts,amssymb}
\usepackage{mathrsfs}
\usepackage{mathtools}
\usepackage{algorithm}
\usepackage{algorithmicx}
\usepackage{algpseudocode}
\usepackage{bm}
\usepackage{multirow}
\usepackage{array}
\usepackage{amssymb}
\usepackage{amsmath}
\usepackage{cite}
\usepackage{url}
\usepackage{xcolor}
\usepackage{cite,graphicx,amsmath,amssymb}
\usepackage[caption=false,font=footnotesize]{subfig}
\usepackage{fancyhdr}
\usepackage{mdwmath}
\usepackage{mdwtab}
\usepackage{caption}
\usepackage{amsthm}
\usepackage{setspace}
\usepackage{bm}
\usepackage{algorithm}
\usepackage{algpseudocode}
\usepackage{mathtools}
\usepackage{dsfont}
\usepackage{bbm}
\newtheorem{remark}{Remark}
\newtheorem{theorem}{Theorem}

\newtheorem{lemma}{Lemma}

\newtheorem{corollary}{Corollary}

\makeatletter
\newcommand{\biggg}{\bBigg@{3}}
\newcommand{\Biggg}{\bBigg@{3.5}}
\makeatother
\def\BibTeX{{\rm B\kern-.05em{\sc i\kern-.025em b}\kern-.08em
    T\kern-.1667em\lower.7ex\hbox{E}\kern-.125emX}}

\begin{document}
\title{Robust Beamforming and Power Allocation for Coherent Cell-Free Massive MIMO with Residual Calibration Errors}
\author{
\author{
Mingjun Sun, Xidong Mu, Shaochuan Wu, Chongjun Ouyang, and Hyundong Shin,~\IEEEmembership{Fellow,~IEEE}
\thanks{Mingjun Sun and Shaochuan Wu are with the School of Electronics and Information Engineering, Harbin Institute of Technology, Harbin 150001, China (e-mail: sunmj@stu.hit.edu.cn; scwu@hit.edu.cn).}
\thanks{Xidong Mu is with the Centre for Wireless Innovation (CWI), School of Electronics, Electrical Engineering and Computer Science, Queen's University Belfast, Belfast, BT3 9DT, U.K. (e-mail: x.mu@qub.ac.uk).}
\thanks{Chongjun Ouyang is with the School of Electronic Engineering and Computer Science, Queen Mary University of London, London, E1 4NS, U.K. (e-mail: c.ouyang@qmul.ac.uk).}
\thanks{Hyundong Shin is with the Department of Electronics and Information Convergence Engineering, Kyung Hee University, Yongin-si, Gyeonggi-do 17104, Republic of Korea (e-mail: hshin@khu.ac.kr).}
}}
\maketitle

\begin{abstract}
This paper investigates robust downlink transmission to tolerate calibration aging in time-division duplex cell-free massive multiple-input multiple-output (CF-mMIMO) systems with residual calibration errors (RCEs).
Unlike existing studies that typically treat RCEs as static impairments, we develop a time-evolving RCE model that characterizes the joint effects of residual phase mismatches, residual carrier frequency offsets, and oscillator phase noise. 
Based on this model, two practical processing architectures are considered:
instantaneous calibrated-channel-based robust beamforming (BF) and statistical
beamformed-channel-based robust power allocation (PA). For both architectures,
tractable achievable rate lower bounds are derived, which explicitly reveal the
impact of calibration aging on coherent combining, BF-gain uncertainty, and inter-user interference.
Using these lower bounds as design metrics, we formulate an effective weighted sum-rate (EWSR) maximization problem over the data transmission interval, so that the resulting BF and PA designs can account for the temporal evolution of RCEs rather than a single calibrated instant. To efficiently solve the resulting problems, a Gauss--Legendre quadrature-based weighted minimum mean square error (WMMSE) optimization framework is developed, where both robust BF and PA are updated in an access point (AP)-block manner with closed-form solutions under per-AP power constraints. Simulation results demonstrate that: i) the proposed algorithms exhibit stable convergence; ii) the proposed robust BF and PA schemes achieve higher EWSR by explicitly accounting for calibration aging than their non-robust counterparts; and iii) robust BF achieves higher spectral efficiency (SE), whereas robust PA provides a more favorable tradeoff between SE and implementation cost in terms of computational complexity and fronthaul overhead.
\end{abstract}


\vspace{-7mm}
\section{Introduction}
Cell-free massive multiple-input multiple-output (CF-mMIMO) leverages a distributed large-aperture architecture to provide high directivity gains while eliminating cell-edge effects, making it a key enabling technology for 6G \cite{111,Hien,MMichail,Emil_making_CF}. To fully unlock the potential of CF systems, coherent downlink beamforming (BF) is indispensable, since the signals transmitted from distributed access points (APs) should be coherently combined at the user side. This, in turn, requires accurate downlink channel state information (CSI) and phase synchronization among APs. However, achieving these two requirements in practical CF systems is highly challenging.

To obtain downlink CSI, time-division duplexing (TDD) is widely adopted to reduce downlink pilot and feedback overhead, where downlink BF is designed from uplink channel estimates by exploiting channel reciprocity \cite{Hien,Ngo_nopilot}. Although the physical propagation channel is reciprocal, the effective uplink and downlink channels are affected by different transmit and receive radio frequency (RF) chains. Hardware components such as amplifiers, filters, and A/D converters introduce unknown amplitude scalings and phase rotations, resulting in channel non-reciprocity \cite{Correctly,intro1}. Consequently, the BF vectors designed from uplink channel estimates may be mismatched with the actual downlink channels. The phase synchronization requirement is even more challenging in CF systems. Specifically, geographically distributed APs are generally driven by independent local oscillators (LOs), and the resulting carrier frequency offsets (CFOs) induce time-varying relative phase rotations across APs. These phase rotations can rapidly impair the coherent superposition of jointly transmitted signals, even when the BF vectors are designed based on accurate downlink CSI \cite{Erik}. Therefore, both reciprocity calibration and frequency synchronization are essential for practical coherent CF-mMIMO transmission \cite{rogalin2014scalable}.

Several recent studies have examined the impact of channel non-reciprocity on CF-mMIMO systems \cite{CF_Non_1,Ohashi2021CFmMIMO_NOMA,URLLC}. In \cite{CF_Non_1}, the downlink achievable rate under conjugate BF was analyzed, showing that AP-side non-reciprocity phase mismatch is a dominant source of performance degradation. This issue was further investigated in non-orthogonal multiple access (NOMA) and ultra-reliable low-latency communication (URLLC) scenarios \cite{Ohashi2021CFmMIMO_NOMA,URLLC}. These studies demonstrate that ignoring imperfect reciprocity may lead to overly optimistic performance evaluation. To mitigate the reciprocity mismatch, various calibration techniques have been developed. Early approaches relied on additional hardware circuits to compensate for the non-reciprocal responses of transceiver chains \cite{hardware_1,hardware_2}. To reduce hardware cost and implementation complexity, over-the-air (OTA) calibration methods based on bidirectional channel measurements were proposed \cite{guillaud2005reciprocity}. In \cite{shepard2012argos}, a relative calibration algorithm was further developed without requiring user-side participation, thereby avoiding additional signaling protocols. For distributed antenna systems with low signal-to-noise ratio (SNR) calibration links, the authors of \cite{rogalin2014scalable} proposed a least-squares (LS)-based method with favorable performance. To address the high computational complexity of LS calibration and its vulnerability to noisy links, a graph neural network (GNN)-based calibration algorithm was proposed in \cite{GNN_RC}. Moreover, several pilot-efficient OTA calibration schemes were investigated in \cite{avalanche,LS3,xu2023spanning} to reduce training overhead and shorten the time gap between calibration and data transmission.

In addition to reciprocity calibration, several works have explicitly addressed the frequency synchronization problem in distributed antenna systems. Existing approaches commonly follow a master--slave architecture, where slave APs estimate and compensate their relative CFOs with respect to a master AP. For example, AirSync was proposed in \cite{AirSync}, where the master AP broadcasts a reference pilot outside the data transmission band, while each slave AP uses a dedicated antenna to track the phase drift and compensate it during data transmission. To improve synchronization generality, AirShare was developed in \cite{AirShare}, where a clock emitter is deployed within the service area and each distributed AP is equipped with a simple circuit to capture the common clock reference. In \cite{rogalin2014scalable}, synchronization pilots were embedded into the TDD frame structure. To support large-scale geographically distributed deployments, a subset of APs was selected as anchors, and graph coloring was employed to assign orthogonal pilots among them, while the remaining APs achieved synchronization by receiving pilots from nearby anchor APs. 
More recently, the authors of \cite{BeamSync} proposed BeamSync, a two-stage protocol for frequency synchronization, where synchronization pilots are beamformed along the dominant inter-AP channel direction and an non-linear LS-based CFO estimator is derived. Nevertheless, the resulting CFO estimation still involves a one-dimensional search over the candidate frequency offsets. To obtain a closed-form estimator, Sync4CT was proposed in \cite{Sync4CT}, which designs a customized orthogonal frequency division multiplexing (OFDM) synchronization signal composed of two identical halves. By exploiting this repeated structure, Sync4CT enables each slave AP to estimate the CFO in closed form and allows multiple slave APs to perform CFO estimation in parallel, thereby reducing both synchronization latency and pilot overhead.

Despite the above progress, practical calibration and synchronization are inevitably imperfect. Receiver noise, weak master--slave links, and finite pilot resources leave residual calibration errors (RCEs) after each calibration update. More importantly, these RCEs do not remain static during the subsequent data transmission interval. In practical 5G-NR-type air interfaces, downlink data transmission usually starts after a non-negligible delay caused by sounding reference signal (SRS)-based channel estimation, scheduling, and link adaptation for modulation and coding scheme (MCS) selection\cite{Access1}. Moreover, in CF networks, fronthaul data exchange and centralized processing introduce additional latency \cite{Access1,CF_fronthaul}. During this delay and the following data transmission phase, residual CFOs accumulate into time-varying phase drifts \cite{AirSync}, while oscillator phase noise (PN) introduces additional random phase fluctuations \cite{phase_noise1,phase_noise2}. 
As a result, the inter-AP coherent combining gain gradually degrades with the calibration age, and downlink BF and power allocation (PA) designs based solely on the instantaneous calibrated channel at the calibration instant become increasingly mismatched with the actual downlink CSI.
This makes robust design that explicitly accounts for RCE evolution highly desirable. Robust BF under non-ideal calibration has been investigated in \cite{RobustCF1,RobustCF2}. However, existing works mainly focus on static channel non-reciprocity, while overlooking the joint impact of residual CFOs and oscillator PN after calibration. The time evolution of RCEs between the calibration instant and the actual data transmission interval also remains insufficiently characterized. Consequently, robust downlink design for coherent CF-mMIMO under time-evolving RCEs remains an open problem. This motivates us to develop a tractable RCE evolution model and design robust BF and PA schemes that explicitly account for calibration aging.
The main contributions of this paper are summarized as follows:

\begin{itemize}
    \item We develop a time-evolving RCE model for TDD CF-mMIMO systems based on the Sync4CT calibration framework. The model jointly captures residual phase mismatches (PMs), residual CFOs, and oscillator PN. Based on this model, we further introduce a coherence factor to characterize the degradation of inter-AP coherent combining over the data transmission interval.
    
    \item We first investigate an instantaneous calibrated-channel-based robust BF architecture. A tractable achievable rate lower bound is derived, revealing the impact of RCE evolution on the desired coherent gain, BF-gain uncertainty, and inter-user interference. Based on this bound, we formulate an effective weighted sum-rate (EWSR) maximization problem and develop a Gauss--Legendre (GL) quadrature-based weighted minimum mean square error (WMMSE) framework, where the BF vectors are efficiently updated in closed form.
    \item We further propose a statistical beamformed-channel-based robust PA architecture to reduce fronthaul overhead and processing latency. In this architecture, each AP locally computes its BF direction, while the central processing unit (CPU) only optimizes the power coefficients. A corresponding achievable rate lower bound and an AP-block WMMSE algorithm are developed, which significantly reduces the CPU-side computational complexity.

    \item Simulation results verify the convergence and effectiveness of the proposed robust algorithms. Both robust schemes substantially outperform non-robust counterparts that ignore RCE evolution, demonstrating the importance of explicitly accounting for calibration aging. Moreover, robust BF achieves higher spectral efficiency, whereas robust PA offers a more latency- and fronthaul-efficient alternative and may even outperform robust BF when practical fronthaul overhead and CPU processing delay are considered.
\end{itemize}


\vspace{-5mm}
\section{System Model and RCE Evolution Model}
\vspace{-2mm}
In this section, we begin by characterizing the main sources of inter-AP asynchronism in distributed coherent transmission, including the CFOs arising from independent LOs at different APs and the PMs caused by imperfect RF chains. Based on the Sync4CT framework \cite{Sync4CT}, we then formulate the RCEs after each update and model their temporal evolution over the subsequent data transmission interval. 

\vspace{-5mm}
\subsection{System Description and Sync4CT-Based RCE Model}

We consider a CF-mMIMO downlink system, where $M$ geographically distributed APs, each equipped with $N_t$ antennas, jointly serve $K$ single-antenna users over the same time-frequency resources. 
The sets of APs and users are denoted by $\mathcal{M}=\{1,\ldots,M\}$ and $\mathcal{K}=\{1,\ldots,K\}$, respectively. Following the user-centric operation principle, user $k$ is served by a subset of nearby APs $\mathcal{M}_k\subseteq\mathcal{M}$, while $\mathcal{K}_m\subseteq\mathcal{K}$ denotes the set of users served by AP $m$. The system operates in TDD mode, where downlink BF is designed based on the channel estimated from uplink pilot training.

Cooperative transmission in such a system relies on the coherent superposition of signals transmitted from different APs. In practice, however, this coherence is impaired by imperfect hardware-induced inter-AP CFOs and PMs.
In this work, we assume that all antennas within the same AP are driven by a common LO, which is shared by the uplink and downlink RF chains \cite{Sync4CT}. Consequently, all antenna branches within the same AP share a common phase reference. Furthermore, the amplitude scaling and phase shifts across different RF chains within each AP are assumed to be locally calibrated via bidirectional pilot signaling \cite{shepard2012argos}, such that their residual impact is sufficiently small. Therefore, for coherent joint transmission, the dominant impairment arises from \emph{inter-AP} mismatch rather than \emph{intra-AP} mismatch. This is a practical assumption because the calibration between APs, constrained by low SNR links, is generally far less effective than the intra-AP calibration \cite{rogalin2014scalable,vieira2017reciprocity}.

The OTA propagation channel is reciprocal and is assumed to follow a block-fading model, i.e., it remains constant within each coherence block. Specifically, the channel between AP $m$ and user $k$ is denoted by $\mathbf{h}_{m,k}\in\mathbb{C}^{N_t\times 1}$ and modeled as
\begin{equation}
    \mathbf{h}_{m,k}\sim \mathcal{CN}\!\left(\mathbf{0},\mathbf{R}_{m,k}\right),
    \label{eq:h_mk}
\end{equation}
where $\mathbf{R}_{m,k}\in \mathbb{C}^{N_t\times N_t}$ is the spatial correlation matrix. The normalized trace $\beta_{m,k}\triangleq \frac{1}{N_t}\operatorname{tr}(\mathbf{R}_{m,k})$ accounts for the large-scale fading coefficient.

\textit{(1) Argos-based intra-AP calibration for PMs:}
Considering the non-ideal RF front-end, the uplink and downlink channels can be expressed as
\begin{subequations}
\begin{align}
\mathbf{h}_{m,k}^{\mathrm{UL}} &= \mathbf{D}_{\mathbf{r}_m}\mathbf{h}_{m,k} t_k, \\
\mathbf{h}_{m,k}^{\mathrm{DL}} &= \mathbf{D}_{\mathbf{t}_m}\mathbf{h}_{m,k} r_k,
\end{align}
\end{subequations}
where $\mathbf{r}_m = [r_{m,1}, \ldots, r_{m,N_t}]^{\mathsf T}\in \mathbb{C}^{N_t \times 1}$ and $\mathbf{t}_m = [t_{m,1}, \ldots, t_{m,N_t}]^{\mathsf T}\in \mathbb{C}^{N_t \times 1}$ denote the receive and transmit RF gain vectors at AP $m$, respectively, and $\mathbf{D}_{\mathbf{r}_m}$ and $\mathbf{D}_{\mathbf{t}_m}$ are the corresponding diagonal matrices constructed from them. Meanwhile, $t_k$ and $r_k$ denote the transmit and receive RF gains at user $k$, respectively.

According to \cite{shepard2012argos}, if each AP performs internal calibration relative to its first antenna, the calibration matrix can be written as
\vspace{-5pt}
\begin{equation}
    \mathbf{C}_{m} = \frac{r_{m,1}}{t_{m,1}}\,\mathbf{D}_{\mathbf{t}_m}\mathbf{D}_{\mathbf{r}_m}^{-1}.
    \label{C_AP}
\end{equation}
By applying $\mathbf{C}_{m}$ to the uplink channel, the actual downlink channel and the calibrated uplink channel satisfy the following relation:
\begin{equation}
\mathbf{h}_{m,k}^{\mathrm{DL}}
=
\underbrace{\mathbf{C}_{m}\mathbf{h}_{m,k}^{\mathrm{UL}}}_{\substack{\text{The calibrated}~\text{uplink channel}}} \cdot
\frac{t_{m,1}r_k}{r_{m,1}t_k}.
\end{equation}
This shows that, for all antennas within the same AP, the actual downlink channel and the calibrated uplink channel differ only by a common complex scalar ${t_{m,1}r_k}/{(r_{m,1}t_k)}$, which does not affect the BF pattern. However, this scalar varies across different APs, and hence an additional inter-AP calibration stage is still required.

\textit{2) Sync4CT-based inter-AP calibration for both CFOs and PMs:}
For notational convenience, define $\mathbf{g}_{m,k} \triangleq \mathbf{C}_{m}\mathbf{h}_{m,k}^{\mathrm{UL}}$ and $\frac{t_{m,1}}{r_{m,1}} \triangleq \rho_m e^{j\nu_m}$, where $\rho_m>0$ and $\nu_m$ denote the corresponding amplitude and phase, respectively.
With both inter-AP CFOs and PMs exist, the received passband signal at user $k$ can be expressed as follows\footnote{Since $r_k/t_k$ is a user-side common scalar and can be locally compensated at the user, it is omitted in the sequel.}:
\begin{equation}\label{yk_1}
y_k = \sum_{i=1}^{K}\sum_{m\in \mathcal{M}_i}
e^{j(2\pi f^{\rm c}_{m}t-\nu_m)}\rho_m \mathbf{g}_{m,k}^{\mathsf{H}}\mathbf{w}_{m,i}s_i + n_k,
\end{equation}
where $s_k$ represents the unit-power data symbol for user $k$, satisfying $\mathbb{E}[|s_k|^2]=1$, $\mathbf{w}_{m,k} \in \mathbb{C}^{N_t \times 1}$ is the BF vector designed at AP $m$ for user $k$, and $n_k \sim \mathcal{CN}(0,\sigma_k^2)$ is the additive complex Gaussian noise at user $k$.
Here, $f^{\rm c}_{m}$ denotes the carrier frequency generated by the LO at AP $m$. Due to oscillator imperfections, the carrier frequencies across different APs are generally not identical. It follows from \eqref{yk_1} that the phase terms $\{2\pi f^{\rm c}_{m}t-\nu_m\}$ across different APs prevent perfectly coherent signal superposition at user side.

Following the master--slave calibration procedure in \cite{Sync4CT,BeamSync}, and without loss of generality, AP~1 is selected as the master AP. The carrier frequency difference and phase offset between AP $m$ and AP~1 are defined as $f_{\mathrm{d},m} = f_{m}^{\mathrm{c}} - f_{1}^{\mathrm{c}}$ and $\nu_{\mathrm{d},m} = \nu_{m} - \nu_{1}$.
As shown in \cite{Sync4CT}, dedicated OFDM signals are designed to estimate $f_{\mathrm{d},m}$ and $\nu_{\mathrm{d},m}$. However, due to receiver noise and the potentially weak master--slave links, the resulting estimates are generally imperfect. We therefore model the residual estimation errors as
\begin{subequations}
\begin{align}
    \tilde{f}_{\mathrm{d},m} &\sim \mathcal{N}(0,\sigma_{f,m}^2), \\
    \tilde{\nu}_{\mathrm{d},m} &\sim \mathcal{N}(0,\sigma_{\nu,m}^2).
\end{align}
\end{subequations}
Accordingly, after imperfect calibration, the received signal for user $k$ can be rewritten as
\begin{equation}\label{yk_cali}
y_k = \varphi_0 \sum_{i=1}^{K}\sum_{m\in \mathcal{M}_i}
e^{j(2\pi \tilde{f}_{\mathrm{d},m}t-\tilde{\nu}_{\mathrm{d},m})}
\rho_m \mathbf{g}_{m,k}^{\mathsf{H}}\mathbf{w}_{m,i}s_i + n_k,
\end{equation}
where $\varphi_0 = e^{j(2\pi f_1^{\rm c} t-\nu_1)}$ is a common phase factor.

\vspace{-3mm}
\subsection{RCE Evolution Modeling}

The above model characterizes the RCEs immediately after the latest Sync4CT update. In practice, however, downlink payload transmission usually starts after a non-negligible delay, during which the RCEs continue to evolve. Therefore, an explicit error evolution model is required. Let $T_{\mathrm{gap}}<T_{\mathrm c}$ denote the pre-transmission delay between the latest calibration update and the beginning of downlink data transmission, where $T_{\mathrm c}$ is the calibration interval. Let $T_{\mathrm s}$ denote the duration of one channel use. The number of discrete-time samples within one calibration interval is then given by $N_{\mathrm c}=\left\lfloor\frac{T_{\mathrm c}}{T_{\mathrm s}}\right\rfloor$, and the starting time instant of the downlink data transmission phase is defined as $n_0\triangleq \left\lceil \frac{T_{\mathrm{gap}}}{T_{\mathrm s}}\right\rceil$.

For each slave AP $m\in\mathcal M\setminus\{1\}$, the RCE at the $n$th time instant is modeled as
\begin{equation}
\phi_m[n]
=
-\tilde{\nu}_{\mathrm d,m}
+
2\pi nT_{\mathrm s}\tilde{f}_{\mathrm d,m}
+
\zeta_{\mathrm d,m}[n].
\label{eq:phi_m_evolution}
\end{equation}
whereas the RCE of the master AP is set to $\phi_1[n]=0$, $n=0,1,\ldots,N_{\mathrm c}$, since it serves as the common phase reference. The first term in \eqref{eq:phi_m_evolution} represents the residual static PM after calibration, mainly caused by RF front-end impairments such as manufacturing imperfections and temperature variations. Since such mismatches typically vary slowly, even over time scales of hours \cite{Erik,Access1}, this term can be regarded as approximately constant over the considered transmission interval. The second term describes the phase drift caused by the residual CFO, which accumulates linearly over time. The third term $\zeta_{\mathrm d,m}[n]$ captures the post-calibration phase fluctuation induced by oscillator PN. Although the PN value at the calibration instant is absorbed into the estimated PM, the oscillator phase continues to fluctuate after calibration, introducing a new time-varying residual phase component during data transmission. Since the master AP is used as the phase reference, $\zeta_{\mathrm d,m}[n]$ denotes the differential PN process between AP $m$ and the master AP, i.e., $\zeta_{\mathrm d,m}[n]=\zeta_m[n]-\zeta_1[n]$.
Following standard oscillator PN models, $\zeta_{\mathrm d,m}[n]$ is
modeled as a discrete-time Wiener process~\cite{phase_noise1,phase_noise2}:
\begin{subequations}
\begin{align}
\zeta_{\mathrm d,m}[0] &= 0,\\
\zeta_{\mathrm d,m}[n]
&=
\zeta_{\mathrm d,m}[n-1]
+
\delta_{\mathrm d,m}[n],
\quad n=1,2,\ldots,N_{\mathrm c},
\end{align}
\end{subequations}
where the increments $\delta_{\mathrm d,m}[n]$ are independent and identically
distributed real Gaussian random variables:
\begin{equation}
\delta_{\mathrm d,m}[n]
\sim
\mathcal N\!\left(0,\sigma_{\zeta,m}^2\right).
\end{equation}
Here,
\begin{equation}
\sigma_{\zeta,m}^2
=
4\pi^2
\left[
\left(f_m^{\mathrm c}\right)^2 c_m
+
\left(f_1^{\mathrm c}\right)^2 c_1
\right]
T_{\mathrm s},
\end{equation}
where $c_m$ is an oscillator-dependent constant associated with AP $m$.
Accordingly, after $N$ time instants, the accumulated phase drift is $\zeta_{d,m}[N]=\sum_{i=1}^{N}\delta_{d,m}[i]$.

Assuming that $\tilde{\nu}_{\mathrm d,m}$, $\tilde{f}_{\mathrm d,m}$, and
$\zeta_{\mathrm d,m}[n]$ are mutually independent, the variance of the total
RCE in \eqref{eq:phi_m_evolution} is given by
\begin{equation}
\label{eq:sigma_phi}
\sigma^2_{\phi,m}[n]
=
\sigma_{\nu,m}^2
+
(2\pi nT_{\mathrm s})^2\sigma_{f,m}^2
+
n\sigma_{\zeta,m}^2.
\end{equation}

\vspace{-3mm}
\section{Instantaneous Calibrated-Channel-Based Robust Beamforming Design}
\vspace{-1mm}
In this section, we investigate centralized robust BF based on the instantaneous calibrated channel. To characterize the effect of RCE evolution on the downlink transmission of CF-mMIMO systems, we first derive a tractable lower bound on the achievable rate and formulate an EWSR maximization problem. Then, by exploiting GL quadrature, we develop an efficient robust BF algorithm for solving the resulting problem.

\vspace{-5mm}
\subsection{Instantaneous Effective-Channel-Based Achievable Rate Lower Bound}
\vspace{-1mm}
In this subsection, we first derive an analytical expression for the achievable rate lower bound and then present several insightful observations on the effect of RCEs. We assume that each user has access to the instantaneous effective channel\footnote{The instantaneous effective channel can be obtained through a downlink BF training \cite{ICSI}.}. 
For notational convenience, define
$b_{m,k,i}\triangleq \mathbf{g}_{m,k}^\mathsf{H}\mathbf{w}_{m,i}$.
At the $n$th time instant after the latest calibration update, the received signal at user $k$ is given by\footnote{The common phase factor $\varphi_0$ is omitted since $|\varphi_0|=1$ and thus it does not affect the achievable-rate analysis. In addition, $\rho_m$ is real-valued, whose impact on coherent combining is much weaker than that of the phase mismatch~\cite{Erik,Sync4CT}. Therefore, $\rho_m$ is omitted in the following derivation for simplicity, while its inclusion is straightforward.}
\begin{equation}\label{y_ins_G_phi}
    y_k[n]
    =
    \sum_{i=1}^{K}\sum_{m\in\mathcal{M}_i}
    e^{j\phi_m[n]} b_{m,k,i} s_i[n]
    + n_k[n].
\end{equation}
Let
$\mathcal{G}\triangleq\{\mathbf g_{m,k},\forall m,k\}$ and
$\boldsymbol\phi\triangleq\{\phi_m[n],\forall m\}$. According to Shannon's theory, the capacity of user $k$ at the $n$th time instant is lower bounded by
\begin{align}\label{pilot_aided_rate}
    &C_k^{\mathrm{iec}}[n] \geq R_k^{\mathrm{iec}}[n]  \\ \nonumber
    &=\mathbb{E}_{\boldsymbol{\phi},\mathcal{G}}
    \!\left[
    \log_2\!
    \left(
    1\!+\!
    \frac{
    \left|\sum_{m\in\mathcal M_k}
    e^{j\phi_m[n]} b_{m,k,k}\right|^2
    }{
    \sum_{i\neq k}
    \left|\sum_{m\in\mathcal M_i}
    e^{j\phi_m[n]} b_{m,k,i}\right|^2
    +
    \sigma_k^2
    }
    \right)\!
    \right].
\end{align}
However, the expectation in $R_k^{\mathrm{iec}}[n]$ makes it
difficult to obtain a closed-form expression and to directly optimize the
BF vectors.
In this work, we consider that the CPU designs the BF vectors based on the instantaneous calibrated
CSI $\mathcal{G}$ and the statistical knowledge of the RCEs\cite{RobustCF1,RobustCF2}, denoted by $\mathcal S_{\phi}$, while the instantaneous RCE realizations $\boldsymbol\phi$ are unavailable at the CPU. Therefore, instead of directly optimizing \eqref{pilot_aided_rate}, we first establish the following information-theoretic relation, which motivates the derivation of a more tractable lower bound for BF design.
\vspace{-3mm}
\begin{lemma}\label{lemma1}
For any given BF design that depends only on
$\mathcal{G}$ and $\mathcal S_{\phi}$, the capacity of user $k$ at the $n$th
time instant with instantaneous effective channel satisfies
\begin{equation}
    C_k^{\mathrm{iec}}[n]
    \geq
    C_k^{\mathrm{sRCE}}[n],
    \label{eq:capacity_ordering_rce}
\end{equation}
where $C_k^{\mathrm{sRCE}}[n]$ denotes the capacity when user $k$ does not know the instantaneous RCE realization and only has access to the RCE-averaged effective channel determined by the statistical knowledge $\mathcal S_{\phi}$ of the RCEs.
\end{lemma}

\begin{proof}
Please see Appendix A.
\end{proof}
\vspace{-6mm}

According to \textbf{Lemma~\ref{lemma1}}, it is sufficient
to derive a tractable lower bound of $C_k^{\mathrm{sRCE}}[n]$ for BF
design.
When only the RCE-averaged effective channel are known at user $k$, the
received signal in \eqref{y_ins_G_phi} can be decomposed as follows:
\begin{equation}
    y_k[n]
    =
    \mathsf{DS}_k[n] s_k[n]
    +
    \mathsf{BU}_k[n] s_k[n]
    +
    \sum_{i\neq k}^K\mathsf{UI}_{k,i}[n] s_i[n]
    +
    n_k[n],
\end{equation}
where
\begin{align}
    \mathsf{DS}_k[n]
    &\triangleq
    \mathbb{E}_{\boldsymbol{\phi}}\!\left[
    \sum_{m\in\mathcal{M}_k}
    e^{j\phi_m[n]} b_{m,k,k}
    \right],
\end{align}
\begin{align}
    \mathsf{BU}_k[n]
    &\triangleq
    \sum_{m\in\mathcal{M}_k}
    e^{j\phi_m[n]} b_{m,k,k}
    -
    \mathsf{DS}_k[n], 
\end{align}
\begin{align}
    \mathsf{UI}_{k,i}[n]
    &\triangleq
    \sum_{m\in\mathcal{M}_i}
    e^{j\phi_m[n]} b_{m,k,i},
    \quad i\neq k \label{UI}.
\end{align}
By construction, $\mathsf{BU}_k[n]$ has zero mean and is therefore uncorrelated with $\mathsf{DS}_k[n]$. Moreover, since the data symbols are mutually independent and independent of the RCEs, the multiuser interference and additive noise are uncorrelated with the desired signal. Therefore, the effective noise is uncorrelated with the desired signal.
By invoking the worst-case uncorrelated Gaussian noise argument~\cite{CF_book,Hien}, $C_k^{\mathrm{sRCE}}[n]$ can be lower bounded as follows:
\begin{align}
\label{Rk_sRCE}
    C_k^{\mathrm{sRCE}}[n] \geq R_{k}^{\mathrm{sRCE}}[n] = 
    \mathbb{E}_{\mathcal{G}}
    \left[
    \log_2
    \left(
    1+
    \Gamma_k[n]
    \right)
    \right],
\end{align}
where $\Gamma_k[n]$ denotes the effective signal-to-interference-plus-noise ratio (SINR) conditioned on $\mathcal{G}$, with its dependence on $\mathcal{G}$ omitted for notational simplicity. Specifically,
\begin{equation}
    \Gamma_k[n]
    \triangleq
    \frac{|\mathsf{DS}_k[n]|^2}
    {
    \mathbb{E}_{\boldsymbol{\phi}}\!\left[|\mathsf{BU}_k[n]|^2\right]
    +
    \sum_{i\neq k}^{K}
    \mathbb{E}_{\boldsymbol{\phi}}\!\left[|\mathsf{UI}_{k,i}[n]|^2\right]
    +
    \sigma_k^2
    }.
    \label{eq:conditional_sinr_srce}
\end{equation}
The outer expectation over $\mathcal{G}$ in \eqref{Rk_sRCE} is used for
ergodic performance evaluation, whereas the BF design only requires
the closed-form expression of the conditional effective SINR $\Gamma_k[n]$. 

To enable a compact reformulation suitable for BF design, we introduce
the following auxiliary definitions:

\emph{1) Coherence factor:}
The coherence factor of AP $m$ at time instant $n$ is defined as
\begin{equation}
    \alpha_m[n]
    \triangleq
    \exp\!\left(
    -\frac{\sigma_{\phi,m}^2[n]}{2}
    \right),
    \quad \forall m\in\mathcal M .
    \label{eq:alpha_def}
\end{equation}

\emph{2) Stacked BF vector:}
The stacked BF vector for user $k$ is defined as
\begin{equation}
    \mathbf w_k
    \triangleq
    \begin{bmatrix}
    \mathbf w_{1,k}^{\mathsf T},
    \mathbf w_{2,k}^{\mathsf T},
    \ldots,
    \mathbf w_{M,k}^{\mathsf T}
    \end{bmatrix}^{\mathsf T}
    \in \mathbb C^{MN_t\times 1},
    \label{eq:stacked_w_def}
\end{equation}
where the blocks corresponding to APs not serving user $k$ are set to zero,
i.e., $\mathbf w_{m,k}=\mathbf 0$, $\forall m\notin\mathcal M_k$.

\emph{3) Effective channel vector:}
For user $k$ at time instant $n$, the effective channel vector is
defined as
\begin{equation}
    \mathbf g_k[n]
    \triangleq
    \begin{bmatrix}
    \alpha_1[n]\mathbf g_{1,k}^{\mathsf T},
    \alpha_2[n]\mathbf g_{2,k}^{\mathsf T},
    \ldots,
    \alpha_M[n]\mathbf g_{M,k}^{\mathsf T}
    \end{bmatrix}^{\mathsf T}
    \in \mathbb C^{MN_t\times 1}.
    \label{eq:effective_channel_def}
\end{equation}

\emph{4) Residual distortion matrix:}
The block-diagonal residual distortion matrix for user $k$ is defined as
\begin{align}
    \mathbf C_k[n]
    \triangleq
    \mathrm{blkdiag} \Big(
    &(1-\alpha_1^2[n])\mathbf g_{1,k}\mathbf g_{1,k}^{\mathsf H},
    \ldots,
    \nonumber\\
    &(1-\alpha_M^2[n])\mathbf g_{M,k}\mathbf g_{M,k}^{\mathsf H}
    \Big).
    \label{eq:Ck_def}
\end{align}

With the above notation, the following theorem gives the closed-form
quadratic expression of $\Gamma_k[n]$.
\vspace{-1mm}
\begin{theorem}
\label{theorem1}
In the presence of RCEs, for any given calibrated CSI realization
$\mathcal G$, the effective SINR $\Gamma_k[n]$ in
\eqref{eq:conditional_sinr_srce} can be equivalently expressed as
\begin{equation}
\label{Gamma_compact}
    \Gamma_k[n]
    =
    \frac{
    \left|\mathbf g_k^{\mathsf H}[n]\mathbf w_k\right|^2
    }
    {
    \sum_{i\neq k}^{K}
    \left|\mathbf g_k^{\mathsf H}[n]\mathbf w_i\right|^2
    +
    \sum_{i=1}^{K}
    \mathbf w_i^{\mathsf H}\mathbf C_k[n]\mathbf w_i
    +
    \sigma_k^2
    }.
\end{equation}
\end{theorem}

\begin{proof}
Please see Appendix B.
\end{proof}

\begin{figure}[!t]
 \centering
\setlength{\abovecaptionskip}{0pt}
\includegraphics[height=0.18\textwidth]{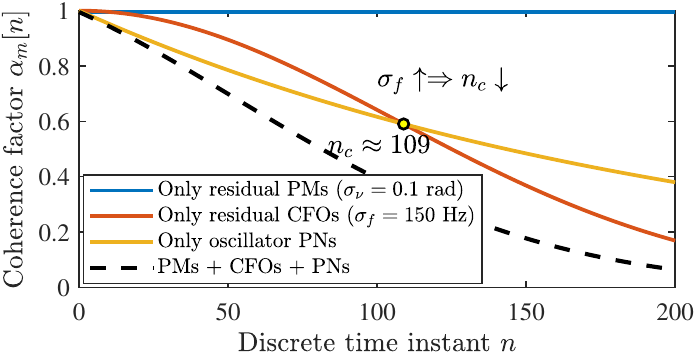}
\caption{Coherence factors versus discrete time instant.}
\label{Coherence_factors}
\vspace{-2mm}
\end{figure}

\vspace{-2mm}
\begin{remark}
Substituting \eqref{eq:sigma_phi} into the coherence factor $\alpha_m[n]$ gives
$
\alpha_m[n]
=
e^{-\frac{\sigma_{\nu,m}^2}{2}}\cdot
e^{-\frac{(2\pi T_{\rm s})^2\sigma_{f,m}^2}{2}n^2}\cdot
e^{-\frac{\sigma_{\zeta,m}^2}{2}n}
$.
This expression shows that $\alpha_m[n]$ decreases monotonically with $n$. Moreover, its decay consists of three distinct components: a constant loss due to the residual static PMs, a Gaussian-type decay in $n$ due to residual-CFO accumulation, and an exponential-type decay in $n$ due to oscillator PN. Therefore, for relatively small $n$, the PN term may dominate the evolution of $\alpha_m[n]$, whereas for sufficiently large $n$, the residual-CFO term becomes dominant. The corresponding crossover index is approximately given by
$
n_{c,m}\triangleq \frac{\sigma_{\zeta,m}^2}{(2\pi T_{\rm s})^2\sigma_{f,m}^2}
$. 
Fig.~\ref{Coherence_factors} provides an intuitive illustration of the
decay behavior of $\alpha_m[n]$. Taking an arbitrary AP $m$ as an example,
we set $\sigma_{\nu,m}=0.1~{\rm rad}$ and $\sigma_{f,m}=150~{\rm Hz}$,
where the subscript $m$ is omitted for notational simplicity. The oscillator
PN parameters are specified in Section~V. Under this setting, the crossover
index is approximately $n_c=109$. As $\sigma_f$ increases, $n_c$ shifts to
a smaller value, indicating that the residual-CFO-induced decay becomes
dominant earlier.
\end{remark}

\vspace{-3mm}
\begin{remark}\label{remark2}
The scalar form of $\Gamma_k[n]$ in \eqref{Gamma_closed_form} enables a direct physical interpretation of the impact of calibration aging.
As $n$ increases, the coherence factors $\alpha_m[n]$ decrease. Hence, the desired coherent component $\left|\sum_{m\in\mathcal M_k}\alpha_m[n]b_{m,k,k}\right|^2$ is generally weakened under coherent BF, while the BF-gain uncertainty term $\sum_{m\in\mathcal M_k}(1-|\alpha_m[n]|^2)|b_{m,k,k}|^2$ increases. Therefore, calibration aging tends to reduce the useful coherent gain and increase the self-interference.
The effect on inter-user interference is more involved. The interference power
consists of the coherent interference component
$\left|\sum_{m\in\mathcal M_i}\alpha_m[n]b_{m,k,i}\right|^2$
and the residual-distortion component
$\sum_{m\in\mathcal M_i}(1-|\alpha_m[n]|^2)|b_{m,k,i}|^2$.
As $\alpha_m[n]$ decreases, these two components may vary in different
directions, and thus the total interference power is not necessarily monotonic
in $n$. Consequently, the strict monotonicity of $\Gamma_k[n]$ or
$R_k^{\mathrm{sRCE}}[n]$ with respect to $n$ cannot be guaranteed. Nevertheless,
in typical coherent BF settings, the loss of the desired coherent gain
is usually dominant, and the achievable rate lower bound generally exhibits a
decreasing trend as calibration ages.
\end{remark}

\vspace{-5mm}
\subsection{Effective Weighted Sum-Rate Maximization Problem}
\vspace{-1mm}
The closed-form $\Gamma_k[n]$ in \eqref{Gamma_compact} provides a
tractable basis for robust BF design under RCEs. Since the coherence factors vary with the calibration age, the conventional weighted sum-rate (WSR) at a single time instant cannot fully characterize the throughput performance over the whole transmission period. Specifically, optimizing only at the
beginning of the data transmission interval may overestimate the coherent
combining gain, whereas optimizing only at the end may result in an overly
conservative design. Therefore, we aim to design BF vectors that balance the
system throughput over the entire data transmission interval. Considering a coherence interval of length $T$, we assume that the calibration interval is matched to the coherence time, i.e., $T_{\rm c}=T$. Accordingly, we define the EWSR utility as
\begin{equation}
\label{WSR_obj}
    U(\{\mathbf w_k\})
    \triangleq
    \frac{1}{N_{\max}}
    \sum_{n=n_0}^{N_{\max}}
    \sum_{k=1}^{K}
    \omega_k \log_2\left(1+\Gamma_k[n]\right),
\end{equation}
where $\omega_k\geq 0$ denotes the priority weight of user $k$ and
$N_{\max}\triangleq \lfloor T/T_{\rm s}\rfloor$ denotes the total discrete
length of the coherence interval.

A direct optimization of \eqref{WSR_obj} requires evaluating the rate at all
sampling instants from $n_0$ to $N_{\max}$, whose number can reach hundreds or
even thousands, thereby causing a considerable computational burden.
To address this issue, we approximate the dense time average by a continuous-time surrogate:
\begin{equation}
\label{eq:WSR_integral}
    \bar U(\{\mathbf w_k\})
    \triangleq
    \frac{1}{T}
    \int_{T_{\rm gap}}^{T}
    \sum_{k=1}^{K}
    \omega_k
    \log_2\left(1+\Gamma_k(t)\right)
    dt.
\end{equation}
Here, $\Gamma_k(t)$ is obtained from \eqref{Gamma_compact} by replacing
$\alpha_m[n]$ with its continuous-time counterpart $\alpha_m(t)=\exp\left(-\frac{\sigma_{\phi,m}^2(t)}{2}\right)$, where $\sigma_{\phi,m}^2(t)=\sigma_{\nu,m}^2 + (2\pi t)^2\sigma_{f,m}^2+\frac{\sigma_{\zeta,m}^2}{T_{\rm s}}t$.

\begin{figure}[!t]
 \centering
\setlength{\abovecaptionskip}{0pt}
\includegraphics[height=0.18\textwidth]{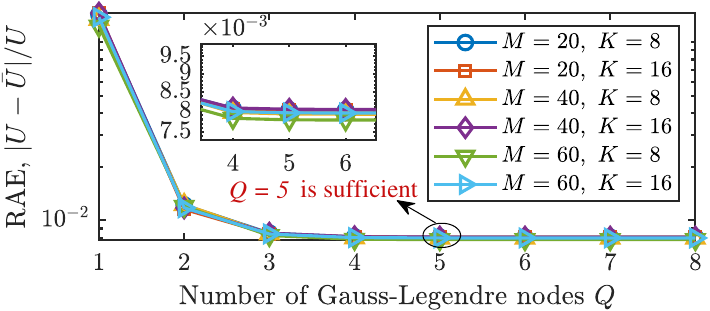}
\caption{RAE of the GL quadrature under different values of $Q$. (The
approximation becomes sufficiently accurate when $Q=5$, where the RAE is below
$0.01$ and can be considered negligible.)}
\label{Gauss_legendre_RAE}
\end{figure}
To efficiently evaluate \eqref{eq:WSR_integral}, we employ the GL quadrature.
For a smooth function $\psi(x)$ on $[-1,1]$, the $Q$-point GL approximation is
given by
\begin{equation}
\label{eq:GLL_standard}
    \int_{-1}^{1}\psi(x)\,dx
    \approx
    \sum_{q=1}^{Q}
    \lambda_q \psi(x_q),
\end{equation}
where $x_q$ is the $q$th root of the $Q$th-order Legendre polynomial and
$\lambda_q$ is the corresponding quadrature weight. By mapping
$[T_{\rm gap},T]$ onto $[-1,1]$, the quadrature nodes in the original time
domain are given by
\begin{equation}
    t_q
    =
    \frac{T-T_{\rm gap}}{2}x_q
    +
    \frac{T+T_{\rm gap}}{2},
    \quad q=1,\ldots,Q .
\end{equation}
Accordingly, \eqref{eq:WSR_integral} can be approximated as
\begin{equation}
\label{eq:WSR_quadrature}
    \bar U(\{\mathbf w_k\})
    \!\approx\!
    \frac{T-T_{\rm gap}}{2T}
    \!\sum_{q=1}^{Q}
    \lambda_q\!
    \sum_{k=1}^{K}
    \omega_k\!
    \log_2\left(1\!+\!\Gamma_k(t_q)\right).
\end{equation}
Compared with the objective in \eqref{WSR_obj}, the
GL-based approximation only requires rate evaluations at $Q$
representative time instants. Therefore, it provides a flexible tradeoff
between approximation accuracy and computational complexity. To further verify the accuracy of the GL approximation, we use the
heuristic maximum-ratio transmission (MRT) BF scheme as an example and compare the GL-based utility with the dense-sampling utility. The relative approximation error (RAE) under different numbers of GL nodes is shown in Fig.~\ref{Gauss_legendre_RAE}. The RAE is defined as the normalized difference
between \eqref{eq:WSR_quadrature} and \eqref{WSR_obj}, i.e., $|U-\bar{U}|/U$.

Based on \eqref{eq:WSR_quadrature}, the robust BF problem is
formulated as
\begin{subequations}
\label{prob:P1}
\begin{align}
\max_{\{\mathbf w_k\}} \quad &
\frac{T-T_{\rm gap}}{2T}
\sum_{q=1}^{Q}
\lambda_q
\sum_{k=1}^{K}
\omega_k
\log_2\!\left(1+\Gamma_k(t_q)\right),
\label{prob:P1_obj}
\\
\text{s.t.}\quad
&
\sum_{k=1}^{K}\|\mathbf w_{m,k}\|_{2}^2\leq P_m,
\quad \forall m\in\mathcal M,
\label{prob:P1_power}
\end{align}
\end{subequations}
where $P_m$ denotes the transmit power budget of AP $m$.

\vspace{-5mm}
\subsection{WMMSE-Based Robust BF Algorithm}
Problem \eqref{prob:P1} is generally non-convex and NP-hard.
To tackle this problem, we adopt the WMMSE framework, which transforms the original WSR maximization problem into an equivalent weighted mean-square error (MSE) minimization problem that is more amenable to block coordinate descent (BCD) updates \cite{wmmse,rwmmse}.

Following the standard WMMSE equivalence, problem \eqref{prob:P1} is equivalent to
\begin{subequations}
\label{prob:P2}
\begin{align}
\min_{\{\mathbf w_k\},\{u_{k,q}\},\{v_{k,q}\}}
&
\sum_{q=1}^{Q}\sum_{k=1}^{K}
\lambda_q\omega_k
\Big(
u_{k,q} e_{k,q}
-\log u_{k,q}
\Big)
\label{prob:P2_obj}\\
\text{s.t.}\quad
&\eqref{prob:P1_power}, \nonumber
\end{align}
\end{subequations}
where $u_{k,q}>0$ and $v_{k,q}\in \mathbb{C}$ are the introduced auxiliary variables. Also, the MSE term $e_{k,q}$ is given by
\begin{align}
e_{k,q}
=\;&
|v_{k,q}|^2
\left(
\sum_{i=1}^{K}
\left|\mathbf g_k^{\mathsf H}(t_q)\mathbf w_i\right|^2
+
\sum_{i=1}^{K}
\mathbf w_i^{\mathsf H}\mathbf C_k(t_q)\mathbf w_i
+
\sigma_k^2
\right)
\nonumber\\
&\;
-2\Re\!\left\{
v_{k,q}^{*}\mathbf g_k^{\mathsf H}(t_q)\mathbf w_k
\right\}
+1 .
\label{eq:mse_kq}
\end{align}

Although problem \eqref{prob:P2} remains jointly non-convex in $\{\mathbf w_k\}$, $\{u_{k,q}\}$, and $\{v_{k,q}\}$, it is convex with respect to each block of variables when the other two blocks are fixed. Therefore, a BCD procedure can be employed.

For fixed $\{\mathbf w_k\}$ and $\{u_{k,q}\}$, the optimal $v_{k,q}^{\star}$ is obtained by setting the first-order derivative of \eqref{eq:mse_kq} with respect to $v_{k,q}^\ast$ to zero, yielding
\begin{equation}
v_{k,q}^{\star}
\!=\!
\frac{\mathbf g_k^{\mathsf H}(t_q)\mathbf w_k}
{\sum_{i=1}^{K}
\left|\mathbf g_k^{\mathsf H}(t_q)\mathbf w_i\right|^2
\!+\!
\sum_{i=1}^{K}
\mathbf w_i^{\mathsf H}\mathbf C_k(t_q)\mathbf w_i
\!+\!
\sigma_k^2 }.
\label{eq:v_update}
\end{equation}
For fixed $\{\mathbf w_k\}$ and $\{v_{k,q}\}$, the optimal $u_{k,q}^{\star}$ is given by
\begin{equation}
u_{k,q}^{\star}=e_{k,q}^{-1}.
\label{eq:u_update}
\end{equation}

For fixed $\{u_{k,q}\}$ and $\{v_{k,q}\}$, retaining only the terms related to
$\{\mathbf w_k\}$ in \eqref{prob:P2_obj} gives the following BF subproblem:
\begin{subequations}
\label{prob:P3}
\begin{align}
\min_{\{\mathbf w_k\}}
\quad &
\sum_{k=1}^{K}
\left(
\mathbf w_k^{\mathsf H}\mathbf A\mathbf w_k
-
2\Re\{\mathbf b_k^{\mathsf H}\mathbf w_k\}
\right)
\label{prob:P3_obj}\\
\text{s.t.}\quad
& \eqref{prob:P1_power}, \nonumber
\label{prob:P3_power}
\end{align}
\end{subequations}
where
\begin{equation}
\mathbf A
\triangleq
\sum_{q=1}^{Q}\sum_{k=1}^{K}
\lambda_q\omega_k u_{k,q}|v_{k,q}|^2
\left(
\mathbf g_k(t_q)\mathbf g_k^{\mathsf H}(t_q)
+
\mathbf C_k(t_q)
\right)
\succeq \mathbf 0,
\label{eq:A_def}
\end{equation}
and
\begin{equation}
\mathbf b_k
\triangleq
\sum_{q=1}^{Q}
\lambda_q\omega_k u_{k,q}v_{k,q}\mathbf g_k(t_q).
\label{eq:b_def}
\end{equation}
Problem \eqref{prob:P3} is a convex quadratically constrained quadratic program (QCQP) problem. In principle, it can be solved via a direct Karush-Kuhn-Tucker (KKT)-based approach. However, such a solution is still computationally demanding. Specifically, the resulting BF admits a closed-form expression involving the inverse of an $MN_t\times MN_t$ matrix. In addition, the $M$ per-AP power constraints introduce $M$ coupled Lagrange multipliers, whose joint optimization leads to a multi-dimensional dual problem with non-negligible computational overhead. To develop a lower-complexity algorithm with closed-form updates, we further exploit the AP-wise separable structure of the constraints and update the BF vectors in an AP-block manner.

To this end, we partition $\mathbf A$ into $M\times M$ blocks and let $\mathbf A_{m,\ell}\in\mathbb C^{N_t\times N_t}$ denote the $(m,\ell)$th block of $\mathbf A$. Similarly, partition $\mathbf b_k$ as
\begin{equation}
\mathbf b_k
=
\begin{bmatrix}
\mathbf b_{1,k}^{\mathsf T},
\mathbf b_{2,k}^{\mathsf T},
\ldots,
\mathbf b_{M,k}^{\mathsf T}
\end{bmatrix}^{\mathsf T}.
\end{equation}
Then the problem in \eqref{prob:P3} can be equivalently rewritten as
\begin{subequations}
\label{prob:P3_block}
\begin{align}
\min_{\{\mathbf w_{m,k}\}} &
\sum_{k=1}^{K}\!
\left[
\sum_{m=1}^{M}\!\sum_{\ell=1}^{M}\!
\mathbf w_{m,k}^\mathsf{H}\mathbf A_{m,\ell}\mathbf w_{\ell,k}
\!-\!
2\Re\!\left\{\!
\sum_{m=1}^{M}\mathbf b_{m,k}^\mathsf{H}\mathbf w_{m,k}\!
\right\}\!
\right]
\label{prob:P3_block_obj}
\\
\text{s.t.}\quad &
\eqref{prob:P1_power}. \nonumber
\end{align}
\end{subequations}
We next adopt an AP-wise BCD strategy to solve \eqref{prob:P3_block}. Since \eqref{prob:P3} is convex and each AP-block subproblem is solved exactly, the AP-block BCD procedure converges to a global optimum of the subproblem.
Specifically, for a given AP $m$, fixing the BF blocks of all other APs,
the AP-block subproblem reduces to
\begin{subequations}
\label{prob:P4}
\begin{align}
\min_{\{\mathbf w_{m,k}\}}
&
\sum_{k\in\mathcal K_m}
\left(
\mathbf w_{m,k}^{\mathsf H}\mathbf A_{m,m}\mathbf w_{m,k}
-
2\Re\{\mathbf d_{m,k}^{\mathsf H}\mathbf w_{m,k}\}
\right)
\label{prob:P4_obj}\\
\text{s.t.}\quad
&
\sum_{k\in\mathcal K_m}\|\mathbf w_{m,k}\|_{2}^2\leq P_m,
\label{prob:P4_power}
\end{align}
\end{subequations}
where
\begin{equation}
\mathbf d_{m,k}
\triangleq
\mathbf b_{m,k}
-
\sum_{\ell\neq m}
\mathbf A_{m,\ell}\mathbf w_{\ell,k}.
\label{eq:d_mk}
\end{equation}

\vspace{-3mm}
Problem \eqref{prob:P4} is a convex quadratic problem with a single quadratic constraint. Its Lagrangian can be written as
\begin{align}
\mathcal L_m
=
\sum_{k\in\mathcal K_m}&
\big(
\mathbf w_{m,k}^\mathsf{H} (\mathbf A_{m,m}+\mu_m\mathbf I_{N_t})\mathbf w_{m,k}\\ \nonumber
&-2\Re\{\mathbf d_{m,k}^\mathsf{H} \mathbf w_{m,k}\}
\big)
-
\mu_m P_m,
\label{eq:Lm}
\end{align}
where $\mu_m\ge 0$ is the Lagrange multiplier. By applying the first-order optimality condition with respect to $\mathbf w_{m,k}^\ast$, we obtain
\begin{equation}
\mathbf w_{m,k}^{\star}(\mu_m)
=
\left(\mathbf A_{m,m}+\mu_m \mathbf I_{N_t}\right)^{-1}\mathbf d_{m,k},
\quad \forall k\in\mathcal K_m.
\label{eq:wmk_mu}
\end{equation}
According to the KKT conditions, if the unconstrained solution with $\mu_m=0$ already satisfies \eqref{prob:P4_power}, then $\mu_m^\star=0$. Otherwise, the optimal multiplier $\mu_m^\star>0$ is the unique root of
\begin{equation}
\phi_m(\mu)
\triangleq
\sum_{k\in\mathcal K_m}
\left\|
\left(\mathbf A_{m,m}+\mu \mathbf I_{N_t}\right)^{-1}\mathbf d_{m,k}
\right\|_2^2
-
P_m
=0.
\label{eq:mu_equation}
\end{equation}
Since $\phi_m(\mu)$ is strictly decreasing for $\mu\ge 0$, the optimal multiplier $\mu_m^\star$ can be efficiently obtained by a one-dimensional bisection search.
By alternating over all APs until convergence, we obtain the solution of problem \eqref{prob:P3}. Compared with a direct solution of \eqref{prob:P3}, the proposed AP-block strategy reduces the matrix inversion dimension from $MN_t$ to $N_t$ and, more importantly, decouples the multi-dimensional dual update into a sequence of simple one-dimensional searches.
The resulting GL-based AP-block WMMSE algorithm is summarized in \textbf{Algorithm~\ref{alg:ap_wmmse}}.  We next analyze its computational complexity and convergence behavior.

\begin{algorithm}[t]
\caption{GL-Based AP-Block WMMSE Algorithm}
\label{alg:ap_wmmse}
\begin{algorithmic}[1]
\State \textbf{Initialization:} Choose a feasible $\{\mathbf w_k^{(0)}\}$ satisfying \eqref{prob:P1_power}; set the iteration index $r=0$.
\Repeat
    \For{$q=1,\ldots,Q$}
        \For{$k=1,\ldots,K$}
            \State Update $v_{k,q}^{(r+1)}$ according to \eqref{eq:v_update}.
            \State Update $u_{k,q}^{(r+1)}$ according to \eqref{eq:u_update}.
        \EndFor
    \EndFor
    \State Form $\mathbf A^{(r+1)}$ and $\{\mathbf b_k^{(r+1)}\}_{k=1}^{K}$ according to \eqref{eq:A_def}, \eqref{eq:b_def}.
    \Repeat
        \For{$m=1,\ldots,M$}
            \State Compute $\mathbf d_{m,k}^{(r+1)}$ for all $k\in\mathcal K_m$ using \eqref{eq:d_mk}.
            \State Update $\mathbf w_{m,k}^{(r+1)}$ for all $k\in\mathcal K_m$ using \eqref{eq:wmk_mu}.
        \EndFor
    \Until{the AP-block BCD for problem \eqref{prob:P3_block} converges}
    \State $r\leftarrow r+1$.
\Until{the objective improvement in \eqref{prob:P1_obj} is below a prescribed tolerance.}
\State \textbf{Output:} $\{\mathbf w_k\}_{k=1}^K$.
\end{algorithmic}
\end{algorithm}

\vspace{-1mm}
\subsubsection{Computational Complexity}
We consider the worst-case computational complexity, where all APs serve all users. In each outer WMMSE iteration, the proposed algorithm mainly consists of three parts: the updates of the auxiliary variables $\{v_{k,q}\}$ and $\{u_{k,q}\}$, the construction of $\mathbf A$ and $\{\mathbf b_k\}_{k=1}^{K}$, and the AP-block BF updates. For each quadrature node $q$ and user $k$, the updates of $v_{k,q}$ and $u_{k,q}$ have a complexity on the order of $\mathcal O(KMN_t)$. Hence, updating all $\{v_{k,q}\}$ and $\{u_{k,q}\}$ over all $Q$ quadrature nodes and $K$ users incurs an overall complexity of $\mathcal O(QK^2MN_t)$.
Before updating the BF vectors, forming $\mathbf A$ and $\{\mathbf b_k\}_{k=1}^K$ according to \eqref{eq:A_def} and \eqref{eq:b_def} requires complexities on the order of $\mathcal O(QKM^2N_t^2)$ and $\mathcal O(QKMN_t)$, respectively. 
Then, the AP-block BF update involves the inner BCD iterations. Based on the update rule in \eqref{eq:wmk_mu}, its overall computational complexity can be expressed as
$\mathcal O\big(
MN_t^3
+
I_{\rm in}
\big(
KM^2N_t^2
+
KMN_t^2
+
KMN_t\log(1/\epsilon)
\big)
\big)$,
where $I_{\rm in}$ denotes the number of inner AP-block BCD iterations, and $\epsilon$ is the prescribed accuracy of the bisection search. The low-complexity implementation of \eqref{eq:wmk_mu}, which leads to the above complexity order, is detailed in the following \textbf{Remark \ref{remark3}}.
Since the complexity associated with updating $\{v_{k,q}\}$ and $\{u_{k,q}\}$ and constructing $\{\mathbf b_k\}_{k=1}^{K}$ is of lower order, the overall complexity is mainly dominated by the construction of $\mathbf A$ and the AP-block BF updates. Therefore, the total computational complexity of the proposed algorithm can be summarized as
$\mathcal O\big(
I_{\rm out}
(
QKM^2N_t^2
+
MN_t^3
+
I_{\rm in}KM^2N_t^2
)
\big)$,
where $I_{\rm out}$ denotes the number of outer WMMSE iterations. 

\vspace{-3mm}
\begin{remark} \label{remark3}
Although the update in \eqref{eq:wmk_mu} appears to require repeated matrix inversions during the bisection search over $\mu_m$, this cost can be substantially reduced since $\mathbf A_{m,m}$ remains fixed throughout the inner AP-block BCD iterations. Hence, $\mathbf A_{m,m}$ can be eigen-decomposed only once as
$
\mathbf A_{m,m}=\mathbf U_m\boldsymbol{\Lambda}_m\mathbf U_m^H
$,
where $\mathbf U_m$ is unitary and
$\boldsymbol{\Lambda}_m=\operatorname{diag}(\lambda_{m,1},\ldots,\lambda_{m,N_t})$. Then, for any trial value of $\mu_m$,
$
(\mathbf A_{m,m}+\mu_m\mathbf I)^{-1}
=
\mathbf U_m(\boldsymbol{\Lambda}_m+\mu_m\mathbf I)^{-1}\mathbf U_m^H
$.
The computation of $\{\mathbf d_{m,k}\}_{k=1}^{K}$ according to \eqref{eq:d_mk} requires $\mathcal O(KMN_t^2)$ operations. By further defining
$
\tilde{\mathbf d}_{m,k}
=
\mathbf U_m^H\mathbf d_{m,k}
=
[\tilde d_{m,k,1},\ldots,\tilde d_{m,k,N_t}]^T
$,
which costs $\mathcal O(KN_t^2)$, the scalar function in \eqref{eq:mu_equation} becomes
\begin{equation}
\phi_m(\mu)
=
\sum_{k\in\mathcal K_m}
\sum_{i=1}^{N_t}
\frac{|\tilde d_{m,k,i}|^2}{(\lambda_{m,i}+\mu)^2}
-
P_m .
\end{equation}
As a result, after the one-time eigenvalue decomposition of $\mathbf A_{m,m}$ with complexity $\mathcal O(N_t^3)$, each bisection step only requires $\mathcal O(KN_t)$ scalar operations. Once the optimal multiplier $\mu_m^\star$ is obtained, the BF vectors can be recovered as
$
\mathbf w_{m,k}^{\star}
=
\mathbf U_m
(\boldsymbol{\Lambda}_m+\mu_m^\star\mathbf I)^{-1}
\tilde{\mathbf d}_{m,k}
$,
which requires $\mathcal O(KN_t^2)$ operations for all users served by AP $m$. Therefore, for one update of AP $m$, the dominant complexity is
$\mathcal O(KMN_t^2+KN_t^2+KN_t\log(1/\epsilon))$, together with the one-time eigenvalue decomposition cost $\mathcal O(N_t^3)$ that is reused throughout the inner BCD iterations.
\end{remark}

\vspace{-3mm}
\subsubsection{Convergence Analysis}
For fixed BF vectors, the auxiliary variables $\{v_{k,q}\}$ and
$\{u_{k,q}\}$ admit closed-form optimal updates. For fixed auxiliary variables,
the BF update subproblem is a convex QCQP, which is solved by exact
minimization over each AP block using the proposed AP-block BCD method.
Accordingly, the objective value of the original problem \eqref{prob:P1} is
non-decreasing after each block update. Since the transmit-power constraints
are bounded, the achievable rate is finite and the objective function is upper
bounded. Hence, the proposed algorithm is guaranteed to converge.

\vspace{-5mm}
\section{Statistical Beamformed-Channel-Based Robust Power Allocation Design}
\vspace{-3mm}
The instantaneous calibrated-channel-based BF design in the previous section adopts a centralized BF architecture, where pilot observations are forwarded from the APs to the CPU for channel estimation and BF design, and the resulting precoded signals are then delivered back to the APs. This fully centralized architecture may incur a considerable fronthaul burden and increase the delay between calibration and data transmission.
To alleviate this issue, we next develop a statistical beamformed-channel-based robust PA scheme. In this scheme, each AP locally computes its precoding direction based on the available calibrated channel estimates and reports only the corresponding statistical beamformed-channel information to the CPU. The CPU then optimizes the network-level PA coefficients and feeds them back to the APs. Compared with the centralized BF design, this distributed architecture reduces the fronthaul overhead, allows the PA coefficients to be updated on a slower time scale, and significantly lowers the computational burden at the CPU. These advantages help shorten the pre-transmission delay $T_{\mathrm{gap}}$, thereby mitigating the SE degradation caused by RCE accumulation.
\vspace{-6mm}
\subsection{Distributed Local BF and Statistical Effective-Channel-Based Achievable Rate Lower Bound}
\vspace{-2mm}
We consider a scenario where the users do not have access to the instantaneous
effective channel gains. Therefore, data detection is performed based on the
statistical effective channels, which can be acquired from long-term channel
statistics.
In the considered distributed processing architecture, each AP locally computes the BF direction based on its available calibrated channel, while the CPU only optimizes the network-level PA coefficients using the statistical beamformed-channel information reported by the APs. Specifically, the transmit BF vector from AP $m$ to user $k$ is denoted as
\begin{equation}
    \mathbf w_{m,k}
    =
    \mu_{m,k}\bar{\mathbf w}_{m,k},
    \label{eq:distributed_beamformer_param}
\end{equation}
where $\bar{\mathbf w}_{m,k}\in\mathbb C^{N_t\times 1}$ denotes the normalized local BF direction satisfying
$\|\bar{\mathbf w}_{m,k}\|_2=1$, and $\mu_{m,k}\geq 0$ is the PA coefficient.

The normalized direction $\bar{\mathbf w}_{m,k}$ can be obtained using different local BF schemes. In this paper, we adopt local linear minimum mean-square error (LMMSE) BF, which is given by~\cite{power_CF}
\begin{equation}
    \bar{\mathbf w}_{m,k}
    \!=\!
    \frac{
    \left(
    \sum_{i\in\mathcal K_m}
    \rho_i
    \mathbf g_{m,i}\mathbf g_{m,i}^{\mathsf H}
    +
    \sigma_{\mathrm{ul}}^2\mathbf I_{N_t}
    \right)^{-1}
    \rho_k \mathbf g_{m,k}
    }{
    \left\|
    \left(
    \sum_{i\in\mathcal K_m}
    \rho_i
    \mathbf g_{m,i}\mathbf g_{m,i}^{\mathsf H}
    +
    \sigma_{\mathrm{ul}}^2\mathbf I_{N_t}
    \right)^{-1}
    \rho_k \mathbf g_{m,k}
    \right\|_2
    },
    \label{eq:local_lmmse_precoder}
\end{equation}
where $\rho_i$ is the uplink transmit power associated with user $i$, and $\sigma_{\mathrm{ul}}^2$ is the uplink receiver noise power.

With the local BF directions determined at the APs, the remaining design variables are the PA coefficients $\{\mu_{m,k}\}$. To characterize the corresponding downlink performance, we apply the use-and-then-forget capacity-bounding technique and obtain the following achievable rate lower bound.
\vspace{-2mm}
\begin{theorem}
\label{theorem:SBC_rate_bound}
When user $k$ has access only to the statistical effective channels, an achievable downlink rate at time instant $n$ is lower bounded by
\begin{equation}
    R_{k}^{\mathrm{sec}}[n]
    =
    \log_2\left(1+\gamma_k^{\mathrm{sec}}[n]\right),
    \label{eq:SBC_rate_bound}
\end{equation}
where the effective SINR $\gamma_k^{\mathrm{sec}}[n]$ is given by
\begin{equation}
\gamma_k^{\mathrm{sec}}[n]
=
\frac{
\left|
\mathbf a_{k,k}^{\mathsf T}[n]\boldsymbol{\mu}_k
\right|^2
}{
\displaystyle
\sum_{i\ne k}
\left|
\mathbf a_{k,i}^{\mathsf T}[n]\boldsymbol{\mu}_i
\right|^2
+
\sum_{i=1}^{K}
\boldsymbol{\mu}_i^{\mathsf T}
\mathbf D_{k,i}[n]
\boldsymbol{\mu}_i
+
\sigma_k^2
}.
\label{eq:SBC_SINR}
\end{equation}
In \eqref{eq:SBC_SINR}, $\boldsymbol{\mu}_i$ denotes the power-allocation vector associated with user $i$, defined as
$\boldsymbol{\mu}_i
\triangleq
[\mu_{1,i},\mu_{2,i},\ldots,\mu_{M,i}]^{\mathsf T}
\in\mathbb R_+^{M\times 1}$,
where $\mu_{m,i}=0$ if AP $m\notin\mathcal M_i$. Moreover, $\mathbf a_{k,i}[n]$ and $\mathbf D_{k,i}[n]$ are respectively defined as
\begin{equation}
\mathbf a_{k,i}[n]
\triangleq
\left[
\alpha_1[n]\bar c_{1,k,i},
\alpha_2[n]\bar c_{2,k,i},
\ldots,
\alpha_M[n]\bar c_{M,k,i}
\right]^{\mathsf T},
\label{eq:stat_a_vector}
\end{equation}
and
\begin{align}
\mathbf D_{k,i}[n]
\triangleq
\operatorname{diag}
\big(&
\Omega_{1,k,i}
-
|\alpha_1[n]|^2|\bar c_{1,k,i}|^2,\\ \nonumber
&\ldots,
\Omega_{M,k,i}
-
|\alpha_M[n]|^2|\bar c_{M,k,i}|^2
\big).
\label{eq:stat_D_matrix}
\end{align}
Here, $\bar c_{m,k,i}$ and $\Omega_{m,k,i}$ denote, respectively, the first- and second-order statistical moments of the local effective scalar channel coefficient 
$c_{m,k,i}\triangleq \mathbf g_{m,k}^{\mathsf H}\bar{\mathbf w}_{m,i}$, i.e.,
$\bar c_{m,k,i}\triangleq \mathbb E_{\mathcal{G}}[c_{m,k,i}]$ and
$\Omega_{m,k,i}\triangleq \mathbb E_{\mathcal{G}}[|c_{m,k,i}|^2]$.
\end{theorem}

\begin{proof}
The proof follows similar steps to that of \textbf{Theorem~\ref{theorem1}} and is therefore omitted for brevity. The only difference is that the expectation is now taken jointly over both the RCEs $\boldsymbol\phi$ and the calibrated channel $\mathcal{G}$.
\end{proof}

\vspace{-5mm}
\subsection{WMMSE-Based Robust PA Algorithm}
By applying the same GL quadrature approximation as in the Section III, the EWSR maximization problem can be formulated in terms of the PA coefficients as follows:
\begin{subequations}
\label{power_control_problem}
\begin{align}
\max_{\{\boldsymbol{\mu}_k\}}
\quad
&
\frac{T-T_{\mathrm{gap}}}{2T}
\sum_{q=1}^{Q}\lambda_q
\sum_{k=1}^{K}
\omega_k
R_{k}^{\mathrm{sec}}(t_q)
\\
\mathrm{s.t.}
\quad
&
\sum_{k=1}^K\mu_{m,k}^2
\le P_m,
\quad \forall m\in\mathcal M,
\label{eq:power_control_problem_b}
\\
&
\mu_{m,k}\ge 0,
\quad \forall m,k.
\label{eq:power_control_problem_c}
\end{align}
\end{subequations}

Since the PA problem has the same sum-logarithmic structure as the centralized BF problem, we again apply the WMMSE transformation. For notational simplicity, we reuse the symbols $u_{k,q}$ and $v_{k,q}$.
The derivation is omitted for brevity, and only the resulting closed-form updates are presented below.
The optimal $u_{k,q}^{\star}$ and $v_{k,q}^{\star}$ can be obtained as follows:
\begin{equation}
v_{k,q}^{\star}
\!=\!
\frac{
\mathbf a_{k,k}^{\mathsf T}(t_q)\boldsymbol{\mu}_k
}{
\sum_{i=1}^{K}
\left|
\mathbf a_{k,i}^{\mathsf T}(t_q)\boldsymbol{\mu}_i
\right|^2
\!+\!
\sum_{i=1}^{K}
\boldsymbol{\mu}_i^{\mathsf T}
\mathbf D_{k,i}(t_q)
\boldsymbol{\mu}_i
\!+\!
\sigma_k^2
}.
\label{eq:v_mu_update}
\end{equation}
\begin{equation}
u_{k,q}^{\star}
\!\!=\!\!
\frac{
\sum_{i=1}^K
\left|
\mathbf a_{k,i}^{\mathsf T}(t_q)\boldsymbol{\mu}_i
\right|^2
\!+\!
\sum_{i=1}^{K}
\boldsymbol{\mu}_i^{\mathsf T}
\mathbf D_{k,i}(t_q)
\boldsymbol{\mu}_i
\!+\!
\sigma_k^2
}{
\sum_{i\ne k}
\left|
\mathbf a_{k,i}^{\mathsf T}(t_q)\boldsymbol{\mu}_i
\right|^2
\!+\!
\sum_{i=1}^{K}
\boldsymbol{\mu}_i^{\mathsf T}
\mathbf D_{k,i}(t_q)
\boldsymbol{\mu}_i
\!+\!
\sigma_k^2
}.
\label{eq:u_mu_update}
\end{equation}

For fixed $\{u_{k,q}\}$ and $\{v_{k,q}\}$, the terms independent of the
PA coefficients can be omitted. The resulting subproblem with
respect to $\{\boldsymbol{\mu}_k\}$ is given by
\begin{subequations}
\label{power_control_problem_BCD}
\begin{align}
\min_{\{\boldsymbol{\mu}_k\}}
\quad
&
\sum_{k=1}^{K}
\left(
\boldsymbol{\mu}_k^{\mathsf T}\mathbf H_k\boldsymbol{\mu}_k
-
2\mathbf p_k^{\mathsf T}\boldsymbol{\mu}_k
\right)
\label{eq:mu_subproblem_a}
\\
\mathrm{s.t.}
\quad
& \eqref{eq:power_control_problem_b}, \eqref{eq:power_control_problem_c}, \nonumber
\end{align}
\end{subequations}
where
\begin{equation}
\mathbf H_k
\!\!\triangleq\!\!
\Re\!\left\{\!
\sum_{q=1}^{Q}\!
\sum_{j=1}^{K}\!
\lambda_q\omega_j u_{j,q}|v_{j,q}|^2
\left(
\mathbf a_{j,k}^{*}(t_q)\mathbf a_{j,k}^{\mathsf T}(t_q)
\!+\!
\mathbf D_{j,k}(t_q)
\right)
\!\!\right\},
\label{eq:H_k_mu_def}
\end{equation}
and
\begin{equation}
\mathbf p_k
\triangleq
\Re\left\{
\sum_{q=1}^{Q}
\lambda_q\omega_k u_{k,q}v_{k,q}^{*}
\mathbf a_{k,k}(t_q)
\right\}.
\label{eq:p_k_mu_def}
\end{equation}

To exploit the per-AP power constraints, we further update the PA
coefficients in an AP-block manner. Let $h_{m,\ell}^{(k)}$ denote the
$(m,\ell)$th element of $\mathbf H_k$, and let $p_m^{(k)}$ denote the $m$th
element of $\mathbf p_k$.
When the PA coefficients associated with APs $\ell\neq m$ are
fixed, the subproblem with respect to $\{\mu_{m,k}\}_{k\in\mathcal K_m}$ becomes
\begin{subequations}
\label{power_control_problem_APm}
\begin{align}
\min_{\{\mu_{m,k}\}}
\quad
&
\sum_{k\in\mathcal K_m}
\left(
h_{m,m}^{(k)}\mu_{m,k}^2
-
2d_{m,k}\mu_{m,k}
\right)
\label{eq:mu_m_subproblem_a}
\\
\mathrm{s.t.}
\quad
&
\sum_{k\in\mathcal K_m}\mu_{m,k}^2
\le P_m,
\label{eq:mu_m_subproblem_b}
\\
&
\mu_{m,k}\ge 0,
\quad \forall k\in\mathcal K_m,
\label{eq:mu_m_subproblem_c}
\end{align}
\end{subequations}
where
\begin{equation}
d_{m,k}
\triangleq
p_m^{(k)}
-
\sum_{\ell\ne m}
h_{m,\ell}^{(k)}\mu_{\ell,k}.
\label{eq:d_mk_mu_def}
\end{equation}

The Lagrangian of \eqref{power_control_problem_APm} is given by
\begin{equation}
\mathcal L_m
\!=\!
\sum_{k\in\mathcal K_m}\!\!
\left(
(h_{m,m}^{(k)}+\xi_m)\mu_{m,k}^2
\!-\!
2d_{m,k}\mu_{m,k}
\right)
-
\xi_m P_m,
\label{eq:mu_m_lagrangian}
\end{equation}
where $\xi_m\ge 0$ is the Lagrange multiplier. By applying the KKT conditions and taking the non-negativity constraint into account, the optimal update of $\mu_{m,k}$ for a given $\xi_m$ is
\begin{equation}
\mu_{m,k}^{\star}(\xi_m)
=
\left[
\frac{d_{m,k}}
{h_{m,m}^{(k)}+\xi_m}
\right]^+,
\quad \forall k\in\mathcal K_m,
\label{eq:mu_mk_closed_form}
\end{equation}
where $[x]^+\triangleq \max\{x,0\}$.
The optimal $\xi_m^\star$ can be efficiently obtained by a one-dimensional bisection search.
By cyclically updating $\{\mu_{m,k}\}_{k\in \mathcal K_m}$ for $m=1,\ldots,M$, the PA subproblem in \eqref{power_control_problem_BCD} can be solved with low complexity. Compared with directly optimizing all power coefficients jointly, the AP-block update only involves scalar closed-form operations and a one-dimensional search for each AP. 
The overall robust PA algorithm follows a procedure similar to the GL-based AP-block WMMSE algorithm in \textbf{Algorithm~\ref{alg:ap_wmmse}}, and hence its pseudocode is omitted to avoid repetition.
We next provide its computational complexity and convergence behavior.


\subsubsection{Computational Complexity}
We also consider the worst-case computational complexity. In each outer WMMSE iteration, the proposed algorithm consists of three main steps: updating the auxiliary variables $\{v_{k,q}\}$ and $\{u_{k,q}\}$, constructing $\{\mathbf H_i\}_{i=1}^{K}$ and $\{\mathbf p_i\}_{i=1}^{K}$, and solving the power-allocation subproblem via AP-block BCD.
For each quadrature node $q$ and user $k$, computing $v_{k,q}$ and $u_{k,q}$ requires a complexity of $\mathcal O(KM)$. Therefore, updating all auxiliary variables over $Q$ quadrature nodes and $K$ users incurs a complexity of $\mathcal O(QK^2M)$. Next, constructing $\{\mathbf H_i\}_{i=1}^{K}$ and $\{\mathbf p_i\}_{i=1}^{K}$ requires complexities of $\mathcal O(QK^2M^2)$ and $\mathcal O(QKM)$, respectively.
For the AP-block BCD update, according to \eqref{eq:mu_mk_closed_form}, the complexity for each AP mainly comes from computing $\{d_{m,k}\}_{k\in\mathcal K}$, which requires $\mathcal O(KM)$ operations, and from finding the optimal Lagrange multiplier $\xi_m$ by bisection, which requires $\mathcal O(K\log(1/\epsilon))$ operations. Therefore, one complete inner AP-block BCD sweep over all APs has complexity $\mathcal O\left(KM^2+KM\log(1/\epsilon)\right)$.
By retaining only the dominant-order terms, the overall computational complexity of the proposed algorithm can be summarized as
$\mathcal O\left(
I_{\rm out}
\left(
QK^2M^2
+
I_{\rm in}KM^2
\right)
\right)
$.

\subsubsection{Convergence Analysis}
The convergence analysis is analogous to that of \textbf{Algorithm~\ref{alg:ap_wmmse}} and is omitted for brevity.

\begin{figure}[!t]
\vspace{-8pt}
    \centering
    \setlength{\abovecaptionskip}{0pt}

    \subfloat[Robust BF]{
        \includegraphics[width=0.44\columnwidth]{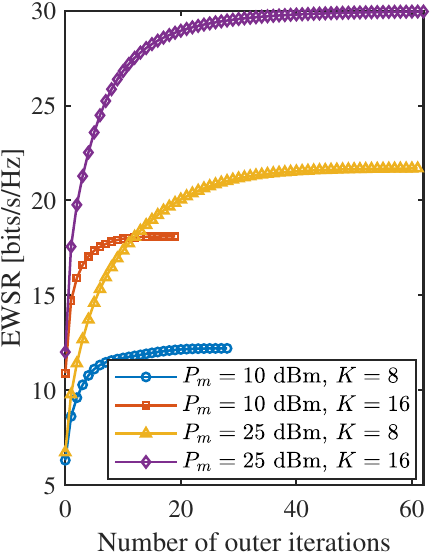}
        \label{fig:con_bf}
    }
    \subfloat[Robust PA]{
        \includegraphics[width=0.44\columnwidth]{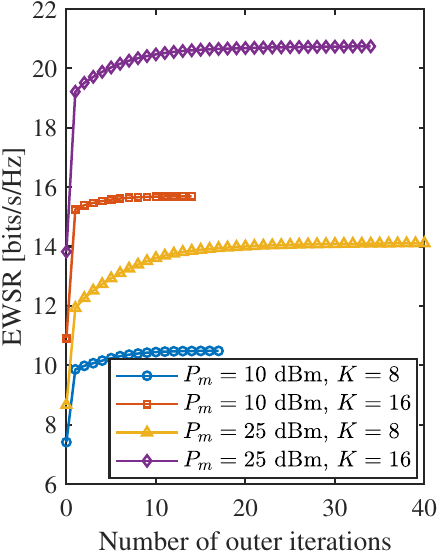}
        \label{fig:con_pa}
    }

    \caption{Convergence curves of the proposed robust algorithms.}
    \label{fig:con_al}
    \vspace{-3pt}
\end{figure}
\vspace{-5mm}
\section{Numerical Results}
\vspace{-2mm}
We consider an urban environment with $M=40$ APs randomly distributed over a 
$1 \times 1~\mathrm{km}^2$ area. The AP and user antenna heights are set to 
$10~\mathrm{m}$ and $1.5~\mathrm{m}$, respectively. The carrier frequency is 3.5 $\mathrm{GHz}$.
According to the 3GPP UMi street-canyon model, the large-scale fading 
between AP $m$ and user $k$ is modeled as\cite{3gpp_tr38901}
\begin{equation}\label{eq:beta_mk}
\beta_{m,k}[\mathrm{dB}] 
= 
-35.3\log_{10}\left(\frac{d_{m,k}}{1~\mathrm{m}}\right)
-34.0
+F_{m,k},
\end{equation}
where $d_{m,k}$ denotes the three-dimensional distance between AP $m$ and user $k$ and $F_{m,k}\sim\mathcal{N}(0,7.82^2)$ models the log-normal shadow fading. To
capture spatial consistency, the shadow-fading terms associated with the same
AP and different users are generated according to an exponential correlation
model, i.e.,
\begin{equation}
\mathbb{E}\{F_{m,k}F_{n,i}\}
=
\begin{cases}
7.82^2
\exp\left(-\dfrac{\delta_{k,i}}{13~\mathrm{m}}\right),
& m=n,\\[2mm]
0,
& m\neq n,
\end{cases}
\end{equation}
where $\delta_{k,i}$ is the horizontal distance between user $k$ and user $i$.
The shadow-fading fields associated with different APs are assumed to be
uncorrelated, which is reasonable for geographically separated APs\cite{Emil_making_CF}. Moreover, the system bandwidth is set to $B=20~\mathrm{MHz}$, the receiver noise figure is $n_f=9~\mathrm{dB}$, and the AWGN power spectral density is $n_0=-174~\mathrm{dBm/Hz}$. The maximum transmit power is $25~\mathrm{dBm}$.
The calibration interval is set to $T_{\rm c}=T=2~\mathrm{ms}$. With a symbol
duration of $T_{\rm s}=10~\mu\mathrm{s}$, the number of temporal samples within
one calibration interval is
$N_{\rm c}=\lfloor T_{\rm c}/T_{\rm s}\rfloor=200$. The oscillator-dependent
coefficient is set to $c_m=10^{-18}$, $\forall m\in\mathcal M$
\cite{JSAC_LO}. Unless otherwise specified, the standard deviations of the
residual CFO and PM errors are set to $\sigma_{f,m}=80~\mathrm{Hz}$ and
$\sigma_{\nu,m}=0.1~\mathrm{rad}$, respectively. 

\begin{figure}[!t]
\vspace{-8pt}
    \centering
    \setlength{\abovecaptionskip}{0pt}

    \subfloat[$K=8$]{
        \includegraphics[width=0.44\columnwidth]{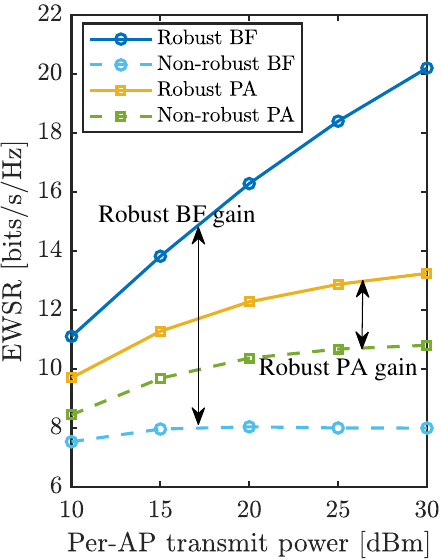}
        \label{power_ewsr_k8}
    }
    \subfloat[$K=16$]{
        \includegraphics[width=0.44\columnwidth]{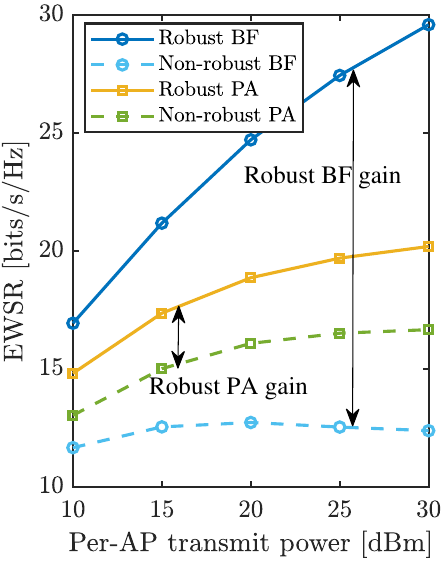}
        \label{power_ewsr_k16}
    }

    \caption{EWSR versus per-AP transmit power for robust and non-robust BF/PA designs.}
    \label{power_ewsr}
    \vspace{-5pt}
\end{figure}
\vspace{-4mm}
\subsection{Convergence Performance}
Figure \ref{fig:con_al} illustrates the convergence behavior of the proposed robust BF and PA algorithms under different per-AP transmit powers and numbers of users. For robust BF, the BF vectors are initialized by MRT. For
robust PA, the PA coefficients are initialized by equal power
allocation among the served users. The convergence tolerances of both the
outer and inner BCD iterations are set to $10^{-4}$.
As shown in Fig.~\ref{fig:con_al}\subref{fig:con_bf} and Fig.~\ref{fig:con_al}\subref{fig:con_pa}, both algorithms exhibit stable and monotonic convergence under all considered settings. Compared with robust BF, robust PA requires fewer outer iterations to converge. Together with the preceding complexity analysis, this demonstrates that the robust PA scheme has not only lower per-iteration computational complexity but also faster practical convergence than the robust BF scheme. Therefore, the proposed robust PA scheme has the potential to reduce the pre-transmission processing delay, thereby shortening the time gap between calibration and data transmission and mitigating the impact of RCE evolution.

\begin{figure}[!t]
 \centering
\setlength{\abovecaptionskip}{0pt}
\includegraphics[height=0.3\textwidth]{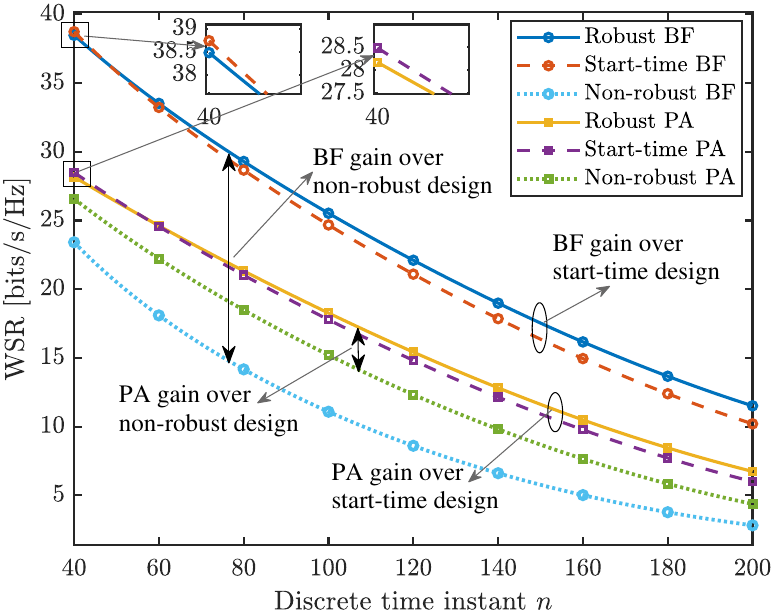}
\caption{WSR evolution of robust, start-time, and non-robust BF/PA designs over the data transmission interval under RCEs.}
\label{Time_instant}
\vspace{-1pt}
\end{figure}
\vspace{-4mm}
\subsection{EWSR Versus Per-AP Transmit Power}
To evaluate the effectiveness of the proposed robust BF and PA algorithms, we
compare them with their non-robust counterparts. The non-robust schemes are
designed under the perfect calibration assumption, while the obtained
BF vectors or PA coefficients are evaluated under the actual
RCE evolution during downlink data transmission.
Fig.~\ref{power_ewsr} shows the EWSR versus the per-AP transmit power for
different numbers of users. As observed from Fig.~\ref{power_ewsr}\subref{power_ewsr_k8} and Fig.~\ref{power_ewsr}\subref{power_ewsr_k16}, the proposed robust BF and PA schemes consistently outperform their non-robust counterparts over the considered power range, which confirms the necessity of explicitly incorporating RCE statistics into the design. The performance gain becomes more evident as the transmit power increases. This is because the system gradually shifts from a noise-limited regime to an interference/error-limited regime, where RCE-induced inter-user interference and BF-gain uncertainty become more pronounced. Since the non-robust schemes ignore these impairments during optimization, their designs become mismatched with the actual downlink channels, leading to reduced desired-signal enhancement and increased inter-user interference and self-distortion caused by BF-gain uncertainty.

It is also observed that the robust gain of BF is larger than that of PA. This
is expected since BF has more spatial degrees of freedom (DoFs) and relies more
strongly on coherent inter-AP signal superposition and interference
suppression. Under RCEs, the effective coherent channel is attenuated by the
coherence factors, and the self-distortion term becomes nonzero. These
effects disturb the coherent combining and interference-suppression structures
designed under the perfect calibration assumption. Consequently, non-robust BF
is more sensitive to RCEs and may even suffer performance
degradation at high transmit power, as shown more clearly in
Fig.~\ref{power_ewsr}\subref{power_ewsr_k16}. In contrast, PA only adjusts real-valued power coefficients based on fixed BF directions, and is therefore less sensitive to inter-AP phase mismatch. This explains the milder degradation of the non-robust PA scheme.

\begin{figure}[!t]
    \vspace{-10pt}
    \centering
    \setlength{\abovecaptionskip}{0pt}

    \subfloat[Robust BF]{
        \includegraphics[width=0.45\columnwidth]{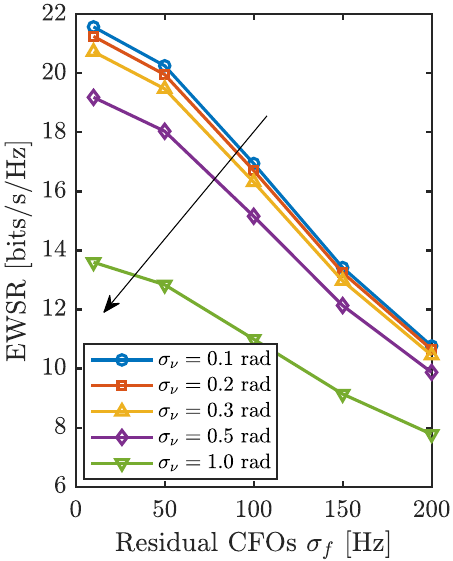}
        \label{sigmaF_ewsr_BF}
    }
    \subfloat[Robust PA]{
        \includegraphics[width=0.45\columnwidth]{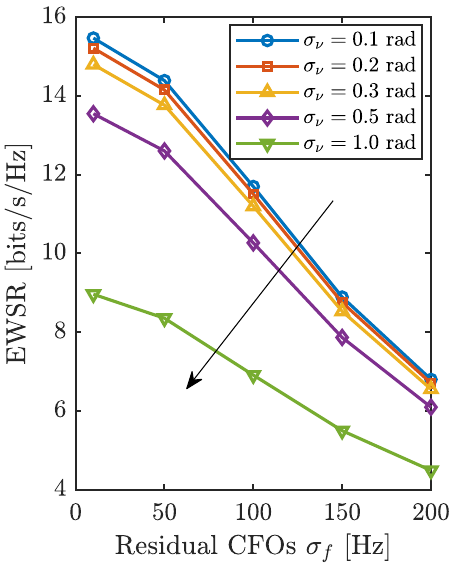}
        \label{sigmaF_ewsr_PA}
    }

    \caption{EWSR versus residual CFOs and PMs.}
    \label{sigmaF_ewsr}
    \vspace{-1pt}
\end{figure}
\vspace{-4mm}
\subsection{WSR Evolution at Discrete Time Instants}
Figure~\ref{Time_instant} illustrates the WSR evolution over the data
transmission interval within one coherence block. The start-time BF and PA
schemes are optimized only for the RCE-affected channel at the
beginning of the data transmission phase.
The results further support the discussion in \textbf{Remark~\ref{remark2}}. As the discrete time
instant increases, the WSR of all schemes gradually decreases, reflecting the
performance loss caused by calibration aging. Nevertheless, the proposed
robust BF and PA schemes consistently outperform their non-robust counterparts
throughout the transmission interval, which confirms the benefit of accounting
for RCE evolution in the design.
The gain over the start-time schemes is relatively moderate. At the beginning
of the transmission interval, the start-time schemes may achieve slightly
higher WSR since they are optimized for that specific time instant. In contrast,
the proposed robust schemes optimize the GL-based time-averaged objective over
multiple representative time instants, and therefore provide better robustness
over the whole interval. Hence, start-time designs can serve as low-complexity
alternatives when pre-transmission processing latency is critical, whereas the
proposed robust designs are more suitable for improving the interval-wise
average performance.

\vspace{-5mm}
\subsection{Impact of Residual CFOs and PMs}
Figure~\ref{sigmaF_ewsr} investigates the impact of residual CFOs and PMs on the EWSR performance of the proposed robust BF and PA schemes.
As shown in Fig.~\ref{sigmaF_ewsr}\subref{sigmaF_ewsr_BF} and Fig.~\ref{sigmaF_ewsr}\subref{sigmaF_ewsr_PA}, the EWSR of both schemes decreases as $\sigma_f$ increases. This is because the residual CFO induces a time-accumulated phase drift during the data transmission phase, which gradually reduces the inter-AP coherent combining gain and increases the effect of inter-user interference and BF-gain uncertainty. Increasing $\sigma_\nu$ also causes an additional EWSR loss. Different from the CFO-induced phase drift, the residual PM introduces a static phase uncertainty after calibration, which results in an overall coherent-gain reduction. It is observed that increasing $\sigma_\nu$ from $0.1$ rad, approximately $5.7^\circ$, to $0.3$ rad, approximately $17.2^\circ$, does not result in a severe EWSR degradation. This level of residual PM accuracy can typically be supported by existing techniques \cite{COTS}.
For the residual CFOs, the performance loss is relatively mild when $\sigma_f$ is below $50$ Hz, but becomes much more pronounced beyond this value. Considering that the CFO mismatch among independently operated and uncalibrated APs may range from hundreds of Hz to tens of kHz in practice \cite{AirShare}, high-accuracy calibration is therefore essential for maintaining reliable coherent transmission.

\begin{figure}[!t]
 \centering
\setlength{\abovecaptionskip}{0pt}
\includegraphics[height=0.3\textwidth]{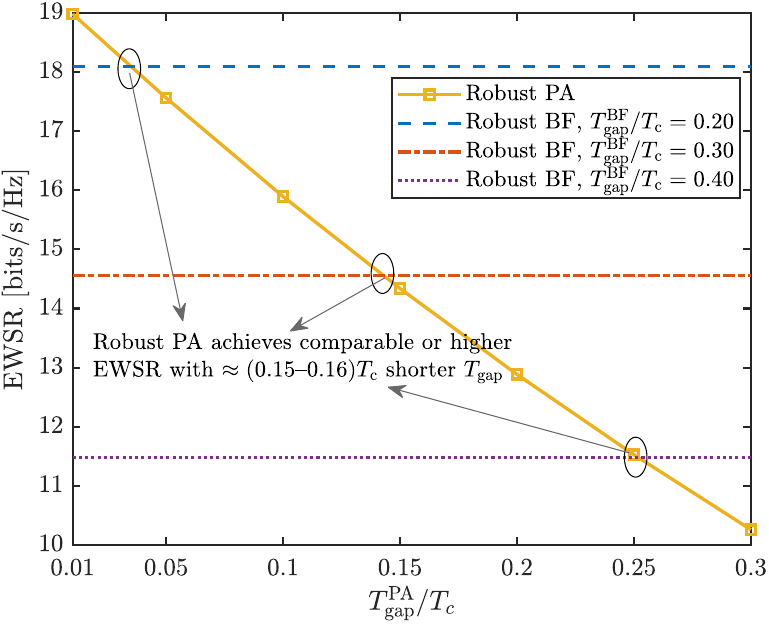}
\caption{EWSR tradeoff between robust PA and robust BF with different normalized pre-transmission delays.}
\label{GapTradeoff}
\end{figure}
\vspace{-4mm}
\subsection{Tradeoff Between Robust BF and PA Under Different Pre-transmission Delays}
Since robust PA incurs lower fronthaul and computational overhead than robust
BF, we further examine the EWSR tradeoff between the two schemes under
different pre-transmission delays. Specifically, robust PA is evaluated with a
varying normalized delay $T_{\mathrm{gap}}^{\mathrm{PA}}/T_c$, while robust BF
is evaluated under several fixed values of
$T_{\mathrm{gap}}^{\mathrm{BF}}/T_{\rm c}$.
As shown in Fig.~\ref{GapTradeoff}, robust PA can achieve comparable, or even
higher, EWSR than robust BF when its pre-transmission delay is approximately
$(0.15\text{--}0.16)T_{\rm c}$ shorter than that of robust BF. This indicates that although robust BF provides higher BF design DoFs, its performance advantage may be offset by the additional processing and fronthaul delay before data transmission. Therefore, robust PA offers a favorable tradeoff between EWSR
performance and implementation latency.

\begin{figure}[!t]
 \centering
\setlength{\abovecaptionskip}{0pt}
\includegraphics[height=0.3\textwidth]{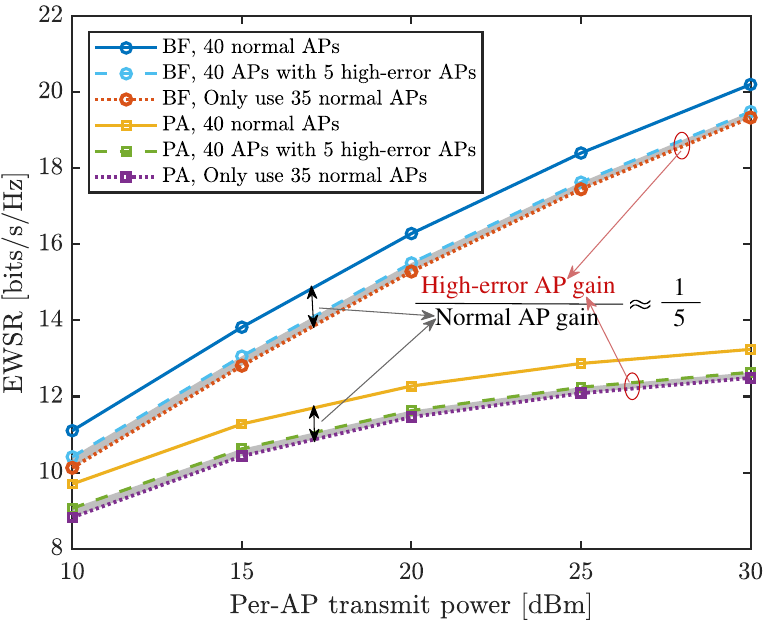}
\caption{EWSR comparison for normal and high-error AP participation. (The shaded regions indicate the EWSR gain obtained by activating the five high-error APs.)}
\label{Good_bad_AP}
\end{figure}

\vspace{-4mm}
\subsection{EWSR Gain from High-Error AP Participation}
Figure~\ref{Good_bad_AP} further investigates the EWSR gain brought by APs with large RCEs. Specifically, we consider a scenario where five APs suffer from high RCEs, with $\sigma_\nu=0.3$ rad and $\sigma_f=200$ Hz, while the normally calibrated APs have $\sigma_\nu=0.1$ rad and $\sigma_f=80$ Hz.
As shown in Fig.~\ref{Good_bad_AP}, activating all APs still provides a slight EWSR improvement compared with directly excluding the five high-error APs. This is expected, since the all-AP case has a larger feasible set and the robust optimization can reduce the transmit power weights assigned to unreliable APs. However, the additional gain provided by the five high-error APs is very limited. For both robust BF and robust PA, this gain is only about one fifth of the gain obtained by activating five normally calibrated APs.
This result indicates that APs with large RCEs may still contribute positively,
but their contribution can be small. Therefore, in practical AP scheduling, the calibration quality should be considered together with channel strength and system overhead. Excluding poorly calibrated APs may reduce fronthaul and scheduling overhead, and this deserves further investigation.
\vspace{-4mm}
\section{Conclusion}
This paper studied robust downlink BF and PA for TDD CF-mMIMO systems under time-evolving RCEs. We first developed an RCE evolution model that captures residual PMs, CFOs, and oscillator PN after each calibration update. Based on this model, we developed two CF implementation architectures, namely a fully centralized robust BF architecture and a partially centralized robust PA architecture. For the former, the CPU designs the BF vectors based on instantaneous calibrated channels. In the latter, the BF directions are locally computed at the APs, while the CPU only performs PA based on statistical beamformed channels. In both architectures, the statistical information of RCEs was explicitly incorporated to account for calibration aging. To support algorithm design, tractable achievable rate lower bounds were derived for both architectures. EWSR maximization problems were then formulated to characterize the system performance over the overall data transmission interval, and two GL-based AP-block WMMSE algorithms were developed to solve them efficiently with low computational complexity. Numerical results confirmed the convergence and effectiveness of the proposed algorithms and demonstrated the importance of explicitly incorporating RCE evolution into downlink design. The results also showed that robust BF achieves higher EWSR gains, whereas robust PA provides a lower-latency implementation. Future work will further consider calibration-aware AP scheduling and AI-enabled low-complexity robust design.

\appendices
\vspace{-4mm}
\section*{Appendix A: Proof of Lemma~1}
Let $S=s_k[n]$ and $Y=y_k[n]$ denote the transmitted symbol and received signal of user $k$ at the $n$th time instant, respectively. Consider a fixed BF policy that depends only on $\mathcal G$ and $\mathcal S_\phi$. In the instantaneous-effective-channel case, the effective channel is specified by $\mathcal G$ and $\phi$, whereas in the RCE-averaged case, it is specified only through $\mathcal G$ and $\mathcal S_\phi$. From an information-theoretic perspective, the corresponding capacities can be expressed as
\begin{align}
    C_k^{\mathrm{iec}}[n]
    &=
    \sup_{p(S)}
    I\!\left(S;Y\mid
    \boldsymbol{\phi},\mathcal{G},\mathcal S_{\phi}\right),
    \label{eq:capacity_iec_mi}
    \\
    C_k^{\mathrm{sRCE}}[n]
    &=
    \sup_{p(S)}
    I\!\left(S;Y\mid
    \mathcal{G},\mathcal S_{\phi}\right),
    \label{eq:capacity_srce_mi}
\end{align}
where $p(S)$ denotes the input distribution of the transmitted data symbol $S$.
Therefore, to prove
$C_k^{\mathrm{iec}}[n]\geq C_k^{\mathrm{sRCE}}[n]$, it is sufficient to show that,
for any admissible $p(S)$,
\begin{equation}
    I\!\left(S;Y\mid
    \boldsymbol{\phi},\mathcal{G},\mathcal S_{\phi}\right)
    \geq
    I\!\left(S;Y\mid
    \mathcal{G},\mathcal S_{\phi}\right).
    \label{eq:mi_ordering_target}
\end{equation}

Since the transmitted data symbol $S$ is independent of the instantaneous RCE realization
$\boldsymbol{\phi}$, we have
\begin{equation}
    I\!\left(S;\boldsymbol{\phi}\mid
    \mathcal{G},\mathcal S_{\phi}\right)=0.
    \label{eq:independent_symbol_rce}
\end{equation}
Using the chain rule of mutual information, we have
\begin{align}
    I\!\left(S;Y,\boldsymbol{\phi}\!\mid\!
    \mathcal{G},\mathcal S_{\phi}\right) 
    \!=\!\!
    I\!\left(S;\boldsymbol{\phi}\!\mid\!
    \mathcal{G},\mathcal S_{\phi}\right)
    \!\!+\!
    I\!\left(S;Y\!\!\mid\!
    \boldsymbol{\phi},\mathcal{G},\mathcal S_{\phi}\right).
    \label{eq:chain_rule_1}
\end{align}
From \eqref{eq:independent_symbol_rce}, \eqref{eq:chain_rule_1} becomes
\begin{equation}
    I\!\left(S;Y,\boldsymbol{\phi}\mid
    \mathcal{G},\mathcal S_{\phi}\right)
    =
    I\!\left(S;Y\mid
    \boldsymbol{\phi},\mathcal{G},\mathcal S_{\phi}\right).
    \label{eq:mi_with_instant_rce}
\end{equation}
On the other hand, applying the chain rule in another order gives
\begin{align}
    I\!\left(S;Y,\boldsymbol{\phi}\!\mid\!
    \mathcal{G},\mathcal S_{\phi}\right) 
    \!=\!
    I\!\left(S;Y\!\!\mid\!
    \mathcal{G},\mathcal S_{\phi}\right)
    \!+\!
    \!I\!\left(S;\boldsymbol{\phi}\!\mid\!\!
    Y,\mathcal{G},\mathcal S_{\phi}\right).
    \label{eq:chain_rule_2}
\end{align}
Since mutual information is nonnegative, we have
\begin{equation}
    I\!\left(S;\boldsymbol{\phi}\mid
    Y,\mathcal{G},\mathcal S_{\phi}\right)\geq 0.
\end{equation}
Therefore,
\begin{equation}
    I\!\left(S;Y,\boldsymbol{\phi}\mid
    \mathcal{G},\mathcal S_{\phi}\right)
    \geq
    I\!\left(S;Y\mid
    \mathcal{G},\mathcal S_{\phi}\right).
    \label{eq:side_information_nonnegative}
\end{equation}
Combining \eqref{eq:mi_with_instant_rce} and
\eqref{eq:side_information_nonnegative}, we obtain
\begin{equation}
    I\!\left(S;Y\mid
    \boldsymbol{\phi},\mathcal{G},\mathcal S_{\phi}\right)
    \geq
    I\!\left(S;Y\mid
    \mathcal{G},\mathcal S_{\phi}\right).
    \label{eq:mi_ordering}
\end{equation}
Since \eqref{eq:mi_ordering} holds for any admissible input distribution
$p(S)$, taking the supremum over $p(S)$ yields
$C_k^{\mathrm{iec}}[n]\geq C_k^{\mathrm{sRCE}}[n]$, which completes the proof.

\vspace{-5mm}
\section*{Appendix B: Proof of Theorem 1}
To obtain a tractable effective SINR for user $k$ at time instant $n$,
we next derive closed-form expressions for the numerator and denominator
of \eqref{eq:conditional_sinr_srce}. For compactness, we first define $z_{k,i}[n]\triangleq\sum_{m\in\mathcal M_i}e^{j\phi_m[n]} b_{m,k,i}$.

\textit{1) Compute $\mathsf{DS}_k[n]$:}
Since $b_{m,k,k}$ is deterministic, by the linearity of expectation we have
\begin{equation}
    \mathsf{DS}_k[n]=\mathbb E_{\boldsymbol{\phi}}\left[z_{k,k}[n]\right]
    =
    \sum_{m\in\mathcal{M}_k}
    \mathbb{E}_{\boldsymbol{\phi}}\!\left[e^{j\phi_m[n]}\right]
    b_{m,k,k}.
\end{equation}
Recall that the RCE $\phi_m[n]$ follows a Gaussian distribution, i.e., $\phi_m[n]\sim\mathcal{N}\left(0,\sigma_{\phi,m}^2[n]\right)$.
Then, by the characteristic function of a Gaussian random variable, it follows that
\begin{equation}
    \mathbb{E}_{\boldsymbol{\phi}}\!\left[e^{j\phi_m[n]}\right]
    =
    \exp\!\left(-\frac{\sigma_{\phi,m}^2[n]}{2}\right)\triangleq\alpha_m[n].
\end{equation}
Therefore, $\mathsf{DS}_k[n]$ can be written in closed form as
\begin{equation}\label{DS_close}
    \mathsf{DS}_k[n]
    =
    \sum_{m\in\mathcal{M}_k}
    \alpha_m[n]\, b_{m,k,k}.
\end{equation}

\textit{2) Compute $\mathbb{E}_{\boldsymbol{\phi}}[|\mathsf{BU}_k[n]|^2]$:}
Since
$\mathsf{BU}_k[n]=z_{k,k}[n]-\mathbb E_{\boldsymbol{\phi}}[z_{k,k}[n]]$,
we obtain $\mathbb E_{\boldsymbol{\phi}}[|\mathsf{BU}_k[n]|^2]=\mathbb E_{\boldsymbol{\phi}}[|z_{k,k}[n]|^2]-\left|\mathbb E_{\boldsymbol{\phi}}[z_{k,k}[n]]\right|^2$.
We first compute $\mathbb E_{\boldsymbol{\phi}}[|z_{k,k}[n]|^2]$:
\begin{align}\label{BU_1}
    &\mathbb E_{\boldsymbol{\phi}}[|z_{k,k}[n]|^2]\\ \nonumber
    &=
    \mathbb{E}_{\boldsymbol{\phi}}\!\left[
    \left(
    \sum_{m\in\mathcal{M}_k}
    e^{j\phi_m[n]} b_{m,k,k}
    \right)
    \left(
    \sum_{\ell\in\mathcal{M}_k}
    e^{-j\phi_\ell[n]} b_{\ell,k,k}^*
    \right)
    \right] \\ \nonumber
    &=
    \sum_{m\in\mathcal{M}_k}|b_{m,k,k}|^2
    +
    \sum_{\substack{m,\ell\in\mathcal{M}_k\\ m\neq \ell}}
    \mathbb{E}_{\boldsymbol{\phi}}[e^{j\phi_m[n]}e^{-j\phi_\ell[n]}]
    b_{m,k,k}b_{\ell,k,k}^*.
\end{align}
For analytical tractability, the RCEs across different APs are assumed to be mutually independent. Hence, for $m\neq \ell$,
\begin{align}
    \mathbb{E}_{\boldsymbol{\phi}}\!\left[
    e^{j\phi_m[n]}e^{-j\phi_\ell[n]}
    \right]
    &=
    \mathbb{E}_{\boldsymbol{\phi}}\!\left[e^{j\phi_m[n]}\right]
    \mathbb{E}_{\boldsymbol{\phi}}\!\left[e^{-j\phi_\ell[n]}\right]  \nonumber\\
    &=
    \alpha_m[n]\alpha_\ell^*[n].
\end{align}
Similarly, we have
\begin{align}\label{BU_2}
    &|\mathbb{E}_{\boldsymbol{\phi}}[z_k[n]]|^2\\ \nonumber
    &=\!
    \sum_{m\in\mathcal{M}_k}
    |\alpha_m[n]|^2 |b_{m,k,k}|^2
    \!+\!
    \sum_{\substack{m,\ell\in\mathcal{M}_k\\ m\neq \ell}}
    \alpha_m[n]\alpha_\ell^*[n]
    b_{m,k,k}b_{\ell,k,k}^*.
\end{align}

By subtracting \eqref{BU_2} from \eqref{BU_1}, the cross terms cancel out, yielding
\begin{equation}\label{BU_close}
    \mathbb{E}_{\boldsymbol{\phi}}[|\mathsf{BU}_k[n]|^2]
    =
    \sum_{m\in\mathcal{M}_k}
    \big(1-|\alpha_m[n]|^2\big)
    |b_{m,k,k}|^2.
\end{equation}

\textit{3) Compute $\mathbb{E}_{\boldsymbol{\phi}}[|\mathsf{UI}_{k,i}[n]|^2]$:}
Since $\mathsf{UI}_{k,i}[n]=z_{k,i}[n]$, it follows from \eqref{BU_1} that
\begin{align}\label{UI_close}
    &\mathbb{E}_{\boldsymbol{\phi}}[|\mathsf{UI}_{k,i}[n]|^2]\\ \nonumber
    &=
    \sum_{m\in\mathcal{M}_i} |b_{m,k,i}|^2
    +
    \sum_{\substack{m,\ell\in\mathcal{M}_i\\ m\neq \ell}}
    \alpha_m[n]\alpha_\ell^*[n]
    b_{m,k,i}b_{\ell,k,i}^* \\ \nonumber
    &=\left|
    \sum_{m\in\mathcal M_i}
    \alpha_m[n]b_{m,k,i}
    \right|^2
    +
    \sum_{m\in\mathcal M_i}
    \left(1-|\alpha_m[n]|^2\right)
    |b_{m,k,i}|^2.
\end{align}

Substituting \eqref{DS_close}, \eqref{BU_close}, and \eqref{UI_close} into the effective SINR in \eqref{eq:conditional_sinr_srce} yields \eqref{Gamma_closed_form}, shown at the top of the next page.
\begin{figure*}[!h]
\begin{align}
\label{Gamma_closed_form}
    \Gamma_k[n]
    =
    \frac{
    \left|
    \sum_{m\in\mathcal M_k}
    \alpha_m[n] b_{m,k,k}
    \right|^2
    }
    {
    \sum_{m\in\mathcal M_k}
    \left(1-|\alpha_m[n]|^2\right)
    |b_{m,k,k}|^2
    +
    \sum_{i\neq k}^{K}
    \left[
    \left|
    \sum_{m\in\mathcal M_i}
    \alpha_m[n] b_{m,k,i}
    \right|^2
    +
    \sum_{m\in\mathcal M_i}
    \left(1-|\alpha_m[n]|^2\right)
    |b_{m,k,i}|^2
    \right]
    +
    \sigma_k^2
    } .
\end{align}
\hrule
\vspace{-6mm}
\end{figure*}
Using the definitions in \eqref{eq:alpha_def}, \eqref{eq:stacked_w_def}, \eqref{eq:effective_channel_def}, and \eqref{eq:Ck_def}, we have:
\begin{equation}\label{def_1}
    \sum_{m\in\mathcal M_i}
    \alpha_m[n]b_{m,k,i}
    =
    \mathbf{g}_k^{\mathsf H}[n]\mathbf{w}_i,
\end{equation}
and
\begin{equation}\label{def_2}
    \sum_{m\in\mathcal M_i}
    \left(1-|\alpha_m[n]|^2\right)
    |b_{m,k,i}|^2
    =
    \mathbf w_i^{\mathsf H}\mathbf C_k[n]\mathbf w_i .
\end{equation}
Substituting \eqref{def_1} and \eqref{def_2} into
\eqref{Gamma_closed_form} yields \eqref{Gamma_compact}, which completes
the proof.

\bibliographystyle{IEEEtran}
\bibliography{mybib}

@inproceedings{shepard2012argos,
  title={Argos: Practical many-antenna base stations},
  author={Shepard, Clayton and Yu, Hang and Anand, Narendra and Li, Erran and Marzetta, Thomas and Yang, Richard and Zhong, Lin},
  booktitle="{Proc. 18th Annu. Int. Conf. Mobile Comput. Netw.}",
  pages={53--64},
  year={2012}
}

@article{rogalin2014scalable,
  title={Scalable synchronization and reciprocity calibration for distributed multiuser {MIMO}},
  author={Rogalin, Ryan and Bursalioglu, Ozgun Y and Papadopoulos, Haralabos and Caire, Giuseppe and Molisch, Andreas F and Michaloliakos, Antonios and Balan, Vlad and Psounis, Konstantinos},
  journal={IEEE Trans. Wireless Commun.},
  volume={13},
  number={4},
  pages={1815--1831},
  year={Apr. 2014},
  publisher={IEEE}
}

@article{vieira2017reciprocity,
  title={Reciprocity calibration for massive {MIMO}: Proposal, modeling, and validation},
  author={Vieira, Joao and Rusek, Fredrik and Edfors, Ove and Malkowsky, Steffen and Liu, Liang and Tufvesson, Fredrik},
  journal={IEEE Trans. Wireless Commun.},
  volume={16},
  number={5},
  pages={3042--3056},
  year={May 2017},
  publisher={IEEE}
}

@article{xu2023spanning,
  title={Spanning tree method for over-the-air channel calibration in 6{G} cell-free massive {MIMO}},
  author={Xu, Shu and Cao, Yang and Li, Chunguo and Wang, Dongming and Yang, Luxi},
  journal={IEEE Trans. Wireless Commun.},
  volume={22},
  number={8},
  pages={5567--5582},
  year={Aug. 2023},
  publisher={IEEE}
}

@misc{3gpp_tr38901,
  author  = {3rd Generation Partnership Project {(3GPP)}},
  title ="{Study on channel model for frequencies from 0.5 to 100 GHz (Release 14), TR 38.901, V19.2.0}",
  year = {2025},
  month = {Dec.}
}

@inproceedings{guillaud2005reciprocity,
  author    = {M. Guillaud and D. T. M. Slock and R. Knopp},
  title     = {A practical method for wireless channel reciprocity exploitation through relative calibration},
  booktitle = {Proc. Int. Symp. Signal Process. Appl. (ISSPA)},
  year      = {2005},
  pages     = {403--406}
}

@inproceedings{avalanche,
  author    = {H. Papadopoulos and O. Y. Bursalioglu and G. Caire},
  title     = {Avalanche: Fast {RF} calibration of massive arrays},
  booktitle = {Proc. IEEE Global Conf. Signal Inf. Process. (GlobalSIP)},
  year      = {2014},
  pages     = {607--611},
}

@ARTICLE{LS3,
  author={Jiang, Xiwen and Decurninge, Alexis and Gopala, Kalyana and Kaltenberger, Florian and Guillaud, Maxime and Slock, Dirk and Deneire, Luc},
  journal={IEEE Trans. Wireless Commun.}, 
  title={A Framework for Over-the-Air Reciprocity Calibration for TDD Massive MIMO Systems}, 
  year={2018},
  month={Sep. },
  volume={17},
  number={9},
  pages={5975-5990},}

@ARTICLE{111,
  author={Mudumbai, Raghuraman and Brown Iii, D. Richard and Madhow, Upamanyu and Poor, H. Vincent},
  journal={IEEE Commun. Mag.}, 
  title={Distributed transmit beamforming: challenges and recent progress}, 
  year={2009},
  month={Feb. },
  volume={47},
  number={2},
  pages={102-110},
 }

@ARTICLE{Hien,
  author={Ngo, Hien Quoc and Ashikhmin, Alexei and Yang, Hong and Larsson, Erik G. and Marzetta, Thomas L.},
  journal={IEEE Trans. Wireless Commun.}, 
  title={Cell-Free Massive {MIMO} Versus Small Cells}, 
  year={2017},
  month={Mar.},
  volume={16},
  number={3},
  pages={1834-1850},
 }

@ARTICLE{MMichail,
  author={Matthaiou, Michail and Yurduseven, Okan and Ngo, Hien Quoc and Morales-Jimenez, David and Cotton, Simon L. and Fusco, Vincent F.},
  journal={IEEE Commun. Mag.}, 
  title={The Road to 6{G}: Ten Physical Layer Challenges for Communications Engineers}, 
  year={2021},
  month={Jan.},
  volume={59},
  number={1},
  pages={64-69},
  }

@ARTICLE{Correctly,
  author={Nissel, Ronald},
  journal={IEEE Commun. Lett.}, 
  title={Correctly Modeling {TX} and {RX} Chain in (Distributed) Massive MIMO—New Fundamental Insights on Coherency}, 
  year={2022},
  month={Oct.},
  volume={26},
  number={10},
  pages={2465-2469},
 }

@ARTICLE{Erik,
  author={Larsson, Erik G.},
  journal={IEEE Trans. Signal Process.}, 
  title={Massive Synchrony in Distributed Antenna Systems}, 
  year={2024},
  month={Jan.},
  volume={72},
  number={},
  pages={855-866},
  }

@ARTICLE{Sync4CT,
  author={Wang, Xi and Xu, Fan and Shi, Qingjiang},
  journal={IEEE J. Sel. Areas Commun.}, 
  title={Sync4CT: Synchronization for Coherent Transmission in Distributed Massive {MIMO}}, 
  year={2026},
  month={Jan.},
  volume={44},
  number={},
  pages={2854-2871},
 }

@ARTICLE{BeamSync,
  author={Kunnath Ganesan, Unnikrishnan and Sarvendranath, Rimalapudi and Larsson, Erik G.},
  journal={IEEE Trans. Wireless Commun.},
  title={BeamSync: Over-the-Air Synchronization for Distributed Massive {MIMO} Systems}, 
  year={2024},
  month={Jul.},
  volume={23},
  number={7},
  pages={6824-6837}}

@ARTICLE{Access1,
  author={Vaghefi, Reza Monir and Palat, Ramesh Chembil and Marzin, Giovanni and Basavaraju, Kiran and Feng, Yiping and Banu, Mihai},
  journal={IEEE Access}, 
  title={Achieving Phase Coherency and Gain Stability in Active Antenna Arrays for Sub-6 {GHz} {FDD} and {TDD} {FD-MIMO}: Challenges and Solutions}, 
  year={2020},
  volume={8},
  number={},
  pages={152680-152696}}

@ARTICLE{CF_fronthaul,
  author={Marsch, Patrick and Fettweis, Gerhard},
  journal={IEEE Trans. Wireless Commun.}, 
  title={Uplink {CoMP} under a Constrained Backhaul and Imperfect Channel Knowledge}, 
  year={2011},
  month={Jun.},
  volume={10},
  number={6},
  pages={1730-1742}
  }

@ARTICLE{phase_noise1,
  author={Bj{\"o}rnson, Emil and Matthaiou, Michail and Debbah, Mérouane},
  journal={IEEE Trans. Wireless Commun.}, 
  title={Massive {MIMO} with Non-Ideal Arbitrary Arrays: Hardware Scaling Laws and Circuit-Aware Design}, 
  year={2015},
  month={Aug.},
  volume={14},
  number={8},
  pages={4353-4368}}

@ARTICLE{phase_noise2,
  author={Pitarokoilis, Antonios and Mohammed, Saif Khan and Larsson, Erik G.},
  journal={IEEE Trans. Wireless Commun.}, 
  title={Uplink Performance of Time-Reversal MRC in Massive MIMO Systems Subject to Phase Noise}, 
  year={2015},
  month={Feb.},
  volume={14},
  number={2},
  pages={711-723},
}

@ARTICLE{intro1,
  author={Raeesi, Orod and Gokceoglu, Ahmet and Zou, Yaning and Bj{\"o}rnson, Emil and Valkama, Mikko},
  journal={IEEE Trans. Commun.}, 
  title={Performance Analysis of Multi-User Massive {MIMO} Downlink Under Channel Non-Reciprocity and Imperfect {CSI}}, 
  year={2018},
  month={Jun.},
  volume={66},
  number={6},
  pages={2456-2471}}

@ARTICLE{wmmse,
  author={Shi, Qingjiang and Razaviyayn, Meisam and Luo, Zhi-Quan and He, Chen},
  journal={IEEE Trans. Signal Process.}, 
  title={An Iteratively Weighted {MMSE} Approach to Distributed Sum-Utility Maximization for a {MIMO} Interfering Broadcast Channel}, 
  year={2011},
  month={Sep.},
  volume={59},
  number={9},
  pages={4331-4340},
}

@ARTICLE{rwmmse,
  author={Zhao, Xiaotong and Lu, Siyuan and Shi, Qingjiang and Luo, Zhi-Quan},
  journal={IEEE Trans. Signal Process.}, 
  title={Rethinking {WMMSE}: Can Its Complexity Scale Linearly With the Number of {BS} Antennas?}, 
  year={2023},
  month={Feb.},
  volume={71},
  number={},
  pages={433-446},
}

@ARTICLE{power_CF,
  author={Chakraborty, Sucharita and Demir, {\"O}zlem Tu\u{g}fe and Bj{\"o}rnson, Emil and Giselsson, Pontus},
  journal={IEEE Open J. Commun. Soc.}, 
  title={Efficient Downlink Power Allocation Algorithms for Cell-Free Massive {MIMO} Systems}, 
  year={2021},
  volume={2},
  number={},
  pages={168-186},
 }

@ARTICLE{ICSI,
  author={Interdonato, Giovanni and Ngo, Hien Quoc and Frenger, Pål and Larsson, Erik G.},
  journal={IEEE Trans. Wireless Commun.}, 
  title={Downlink Training in Cell-Free Massive {MIMO}: A Blessing in Disguise}, 
  year={Nov. 2019},
  volume={18},
  number={11},
  pages={5153-5169},
  }

@ARTICLE{RobustCF1,
  author={Han, Shengqian and Yang, Chenyang and Wang, Gang and Zhu, Dalin and Lei, Ming},
  journal={IEEE Trans. Commun.}, 
  title={Coordinated Multi-Point Transmission Strategies for {TDD} Systems with Non-Ideal Channel Reciprocity}, 
  year={Oct. 2013},
  volume={61},
  number={10},
  pages={4256-4270},
 }

@article{CF_book,
  author    = {{\"O}zlem Tu{\u{g}}fe Demir and Emil Bj{\"o}rnson and Luca Sanguinetti},
  title     = {Foundations of User-Centric Cell-Free Massive {MIMO}},
  journal   = {Found. Trends Signal Process.},
  volume    = {14},
  number    = {3--4},
  pages     = {162--472},
  year      = {Jan. 2021},
}

@ARTICLE{JSAC_LO,
  author={Zheng, Jiakang and Zhang, Jiayi and Cheng, Julian and Leung, Victor C. M. and Ng, Derrick Wing Kwan and Ai, Bo},
  journal={IEEE J. Sel. Areas Commun.}, 
  title={Asynchronous Cell-Free Massive {MIMO} With Rate-Splitting}, 
  year={May 2023},
  volume={41},
  number={5},
  pages={1366-1382}
  }

@ARTICLE{Emil_making_CF,
  author={Bj{\"o}rnson, Emil and Sanguinetti, Luca},
  journal={IEEE Trans. Wireless Commun.}, 
  title={Making Cell-Free Massive {MIMO} Competitive With {MMSE} Processing and Centralized Implementation}, 
  year={Jan. 2020},
  volume={19},
  number={1},
  pages={77-90}
}

@ARTICLE{COTS,
  author={Cao, Yang and Wang, Pan and Zheng, Kang and Liang, Xianghu and Liu, Dongjie and Lou, Mengting and Jin, Jing and Wang, Qixing and Wang, Dongming and Huang, Yongming and You, Xiaohu and Wang, Jiangzhou},
  journal={IEEE J. Sel. Areas Commun.}, 
  title={Experimental Performance Evaluation of Cell-Free Massive {MIMO} Systems Using {COTS} {RRU} With {OTA} Reciprocity Calibration and Phase Synchronization}, 
  year={Jun. 2023},
  volume={41},
  number={6},
  pages={1620-1634},
}

@ARTICLE{Ngo_nopilot,
  author={Ngo, Hien Quoc and Larsson, Erik G.},
  journal={IEEE Trans. Wireless Commun.}, 
  title={No Downlink Pilots Are Needed in {TDD} Massive {MIMO}}, 
  year={May 2017},
  volume={16},
  number={5},
  pages={2921-2935},
  }

@ARTICLE{CF_Non_1,
  author={Morte Palacios, Jorge and Raeesi, Orod and Gokceoglu, Ahmet and Valkama, Mikko},
  journal={IEEE Wireless Commun. Lett.}, 
  title={Impact of Channel Non-Reciprocity in Cell-Free Massive {MIMO}}, 
  year={Mar. 2020},
  volume={9},
  number={3},
  pages={344-348},
  }

@article{Ohashi2021CFmMIMO_NOMA,
  author  = {Ohashi, Aline A. and Costa, Daniel Benevides da and Fernandes, Andr{\'e} L. P. and Monteiro, Waldeir and Failache, Ruan and Cavalcante, Andr{\'e} Mendes and Costa, Jo{\~a}o C. W. A.},
  title   = {Cell-Free Massive {MIMO}-{NOMA} Systems With Imperfect {SIC} and Non-Reciprocal Channels},
  journal = {IEEE Wireless Commun. Lett.},
  year    = {Jun. 2021},
  volume  = {10},
  number  = {6},
  pages   = {1329--1333},
}

@ARTICLE{URLLC,
  author={Fang, Hao and Hu, Han and Zhang, Yao and Qiao, Xu and Yang, Longxiang and Zhu, Hongbo},
  journal={IEEE Wireless Commun. Lett.}, 
  title={Cell-Free Massive {MIMO} for {URLLC} With Imperfect Channel Reciprocity}, 
  year={Jul. 2024},
  volume={13},
  number={7},
  pages={1888-1892}
 }

@ARTICLE{AirSync,
  author={Balan, Horia Vlad and Rogalin, Ryan and Michaloliakos, Antonios and Psounis, Konstantinos and Caire, Giuseppe},
  journal={IEEE/ACM Trans. Netw.},
  title={AirSync: Enabling Distributed Multiuser {MIMO} With Full Spatial Multiplexing},
  year={2013},
  month={Dec.},
  volume={21},
  number={6},
  pages={1681--1695}
}

@ARTICLE{GNN_RC,
  author={Sun, Mingjun and Mu, Xidong and Wu, Shaochuan and Ouyang, Chongjun and Liu, Yuanwei},
  journal={IEEE Wireless Commun. Lett.}, 
  title={{GNN}-aided Over-the-Air Reciprocity Calibration for Cell-Free Massive {MIMO} Systems}, 
  year={May 2026, early access},
  volume={},
  number={},
  pages={1-1},
 }

@inproceedings{hardware_1,
  author    = {Bourdoux, Andr{\'e} and Come, Boris and Khaled, Nadia},
  title     = {Non-reciprocal transceivers in {OFDM/SDMA} systems: Impact and mitigation},
  booktitle = {Proc. IEEE Radio Wireless Conf. (RAWCON)},
  address   = {Boston, MA, USA},
  year      = {Aug. 2003},
  pages     = {183--186}
}

@inproceedings{hardware_2,
  author    = {Jungnickel, V. and Kr{\"u}ger, V. and Istoc, G. and Haustein, T. and von Helmolt, C.},
  title     = {A {MIMO} system with reciprocal transceivers for the time-division duplex mode},
  booktitle = {Proc. IEEE Antennas Propag. Soc. Int. Symp.},
  year      = {Jun. 2004},
  volume    = {2},
  pages     = {1267--1270}
}

@inproceedings{AirShare,
  author    = {Abari, Omid and Rahul, Hariharan and Katabi, Dina and Pant, Mondira},
  title     = {{AirShare}: Distributed coherent transmission made seamless},
  booktitle = {Proc. IEEE Conf. Comput. Commun. (INFOCOM)},
  year      = {Apr. 2015},
  pages     = {1742--1750},
}

@ARTICLE{RobustCF2,
  author={Minasian, Arin and Shahbazpanahi, Shahram and Adve, Raviraj S.},
  journal={IEEE Tran. Commun.}, 
  title={Distributed Massive {MIMO} Systems With Non-Reciprocal Channels: Impacts and Robust Beamforming}, 
  year={Nov. 2018},
  volume={66},
  number={11},
  pages={5261-5277},
  }
\end{document}